\documentclass [11pt,leqno]{amsart}
\usepackage[english]{babel}
\usepackage{amssymb}
\usepackage{cancel,graphicx}
\usepackage{hyperref}
\usepackage{slashed}

\title[Metrics and Causality on Moyal planes]{Metrics and causality on Moyal planes}

\author[N. Franco and J.-C. Wallet]{Nicolas Franco and Jean-Christophe Wallet}

\address[]{Nicolas Franco, Namur Center for Complex Systems (naXys) \& Department of Mathematics, University of Namur, Rue de Bruxelles 61, B-5000 Namur, Belgium}
\email{nicolas.franco@math.unamur.be}

\address[]{Jean-Christophe Wallet, Laboratoire de Physique Th\'eorique, B\^at.\ 210\\
CNRS and Universit\'e Paris-Sud 11, F-91405 Orsay Cedex, France}
\email{jean-christophe.wallet@th.u-psud.fr}

\dedicatory{Dedicated to the memory of Raymond Stora.}
\date{July 2015}
\subjclass[2010]{Primary: 58B34; Secondary: 54E35; 53C50; 54F05.}
\keywords{Noncommutative geometry, spectral distance, causal structures, Moyal spaces, quantum locally compact spaces}
\thanks{Work partially supported by the CNRS PEPS/PTI " Noncommutative metric geometry: From Monge to Higgs"}

\numberwithin{equation}{section}
 
\allowdisplaybreaks[4]
\numberwithin{equation}{section}
\newtheorem{proposition}{Proposition}[section]
\newtheorem{theorem}{Theorem}[section]

\newtheorem{definition}{Definition}[section]

\theoremstyle{definition}
\newtheorem{remark}{Remark}[section]

\newcommand\del{{\partial}}
\newcommand\delbar{{{\bar{\partial}}}}
\newcommand\bbone{{ \mathbb{I}}}

\numberwithin{equation}{section}

\begin{document}

\begin{abstract}
Metrics structures stemming from the Connes distance promote Moyal planes to the status of quantum metric spaces. We discuss this aspect in the light of recent developments, emphasizing the role of Moyal planes as representative examples of a recently introduced notion of quantum (noncommutative) locally compact space. We move then to the framework of Lorentzian noncommutative geometry and we examine the possibility of defining a notion of causality on Moyal plane, which is somewhat controversial in the area of mathematical physics. We show the actual existence of causal relations between the elements of a particular class of pure (coherent) states on Moyal plane with related causal structure similar to the one of the usual Minkowski space, up to the notion of locality.
\end{abstract}

\maketitle

\vspace*{-15pt}

\tableofcontents
\setlength{\parskip}{4pt}

\section{Introduction.}\label{introduction}

\subsection{Noncommutative metric geometry.}

The notion of noncommutative metric geometry can be viewed in some sense as a construction of noncommutative analogs of algebras of Lipschitz functions on metric spaces. Actually, the concept of noncommutative metric space can be traced back to \cite{CONNES1}, \cite{CONNES} where it was realized that spectral triples encode a metric structure on the ``noncommutative space'' described by the C*-algebra of the triple. This latter involves a natural seminorm on the algebra from which one can define a metric on the state space of the algebra. This metric is often referred to as the Connes distance or spectral distance, a terminology that we will use in this paper. It can be viewed as a natural noncommutative analog of the geodesic distance. This can be realized by considering for instance the case of a finite dimensional compact Riemann (spin) manifold exactly described by its by now standard spectral triple. Then, the related Connes distance between any 2 points (pure states) coincides with the geodesic distance between these points. It turns out that the above extension of a metric on a space to its set of probability measures already appeared many years ago in the context of optimal transport for compact metric spaces \cite{kantrub} and generalized to non compact (complete) metric spaces in \cite{wasser1}, \cite{dobrush}, related to the Wasserstein distance. Going back to the commutative example given above, it can be shown that the Connes distance between non pure states is the Wasserstein distance of order 1 between the corresponding probability distributions. We will not explore the aspects related to the Wasserstein distance in this paper.\par

Developments that followed the observation of \cite{CONNES1} focused on the notion of compact noncommutative metric spaces introduced in \cite{Rieffel11}, \cite{Rieffel1} (see also \cite{Rieffel2a}, \cite{Rieffel3}). They gave rise to interesting approximations by matrix algebras of commutative as well as noncommutative spaces pertaining to the physics literature \cite{Rieffel2b}, \cite{latrem-tore}. In order to make connection with a terminology sometimes used in the physics as well as (recent) mathematics litterature, we will call a noncommutative metric space as quantum metric space. The basic observation underlying these developments is that one can endow the state space of a C*-algebra $\mathbb{A}$ with a metric through a suitable choice of a densely defined seminorm $\ell$ on $\mathbb{A}$ which can be viewed as a mere generalization of the usual Lipschitz seminorm. Consider for instance a commutative compact metric space $X$ with metric $\rho$. It is known that $\rho$ is determined from the Lipschitz seminorm $\ell_\rho$ on the commutative unital algebra $\mathbb{A}=C(X)$
\begin{equation}
\ell_\rho(f):=\sup_{ x,y\in X, x\ne y}{{\vert f(x)-f(y) \vert}\over{\rho(x,y)}},
\end{equation}
for any $f\in C(X)$, by the relation
\begin{equation}
\rho(x,y)=\sup\{\vert f(x)-f(y)\vert;\  f\in C(X),\ \ell_\rho(f)\le1\}.\label{metricrho}
\end{equation}
The relation \eqref{metricrho} can be extended to the space of probability measures on $X$, ${\mathfrak{S}}(C(X))$, leading to the Kantorovitch distance. For any $\mu_1,\mu_2\in{\mathfrak{S}}(C(X))$, it is given by
\begin{equation}
\rho(\mu_1,\mu_2)=\sup\{\vert \mu_1(f)-\mu_2(f)\vert;\ f\in C(X),\ \ell_\rho(f)\le1\}
\end{equation}
and is known to metrize the w*-topology on ${\mathfrak{S}}(C(X))$ \cite{kantrub}. This observation supports a natural extension to the case of a noncommutative unital algebra giving rise to the theory of compact quantum metric spaces as first proposed by  Rieffel in \cite{Rieffel11}. This later quantum space can be defined from a pair $(\mathfrak{A},\ L)$ where $\mathfrak{A}$ is an order-unit space and $L$ is a densely defined semi-norm on $\mathfrak{A}$ such that the distance on $\mathfrak{S}(\mathfrak{A})$, the space of states of $\mathfrak{A}$, defined for any $\omega_1,\omega_2\in\mathfrak{S}(\mathfrak{A})$ by
\begin{equation}
d_L(\omega_1,\omega_2)=\sup\{|\omega_1(a)-\omega_2(a) |,\ a\in\mathfrak{A},\  L(a)\le1 \}\label{absdistance}
\end{equation}
yields a finite diameter for $\mathfrak{S}(\mathfrak{A})$ and metrizes the w*-topology on $\mathfrak{S}(\mathfrak{A})$. \par 

In the following, we will use the fact that any order unit space can be realized as a (real) linear subspace of self-adjoint operators on some Hilbert space including a unit to pass from the above framework to the one of the spectral triples. We now introduce definitions that will be used in the sequel. Let $\mathcal{B}(H)$ and $\mathcal{K}(H)$ be respectively the C*-algebra of bounded operators on a separable Hilbert space ${{H}}$ and its ideal (C*-subalgebra) of compact operators on $H$. The domain and spectrum of any operator $T$ will be denoted respectively by $\text{Dom}(T)$ and $\text{Spec}(T)$. It will be convenient to define a spectral triple as follows:
\begin{definition}\label{spectraltriple} 
A spectral triple is defined by the data $\mathfrak{{X}}_D=(\mathbb{A},\ \pi,\ {{H}},\ D)$ in which: i) $\mathbb{A}$ is an involutive algebra with $\pi$ a faithful $*$-representation of $\mathbb{A}$ on $\mathcal{B}(H)$, ii) $D$ is a self-adjoint (possibly unbounded) operator with dense domain satisfying $\pi(a)\text{Dom}(D) \subset\text{Dom}(D)$, $\forall a\in\mathbb{A}$, iii) $\forall a\in\mathbb{A}$, $[D,\pi(a)]\in\mathcal{B}(H)$, iv) $\forall a\in\mathbb{A}$, $\forall\lambda\notin \text{Spec}(D)$, $\pi(a)(D-\lambda)^{-1}\in\mathcal{K}(H)$.
\end{definition}
\begin{remark}\label{precstar-states}
Note that our definition of a spectral triple specifies explicitely the representation for further convenience. The triple is called unital (resp. non unital) when $\mathbb{A}$ is (resp. is not) unital. A state $\omega$ on a C*-algebra $\mathbb{A}$ is defined as usual as a positive linear map $\omega:\mathbb{A}\to\mathbb{C}$ with $||\omega||=1$. We denote by $\mathfrak{S}(\mathbb{A})$ the space of states. In the following, we will sometimes consider cases where $\mathbb{A}$ is a pre-C* algebra, for which the notion of state is still meaningfull. Indeed, let $\bar{\mathbb{A}}$ be the C*-completion of $\mathbb{A}$. Then, the restriction to $\mathbb{A}$ of any state $\omega$ on $\bar{\mathbb{A}}$ defines a unique positive linear map $\omega_{\vert\mathbb{A}}:\mathbb{A}\to\mathbb{C}$ with $||\omega||=1$ while any positive linear map $\omega:\mathbb{A}\to\mathbb{C}$ with norm 1 extends to a state on $\bar{\mathbb{A}}$ by continuity.
\end{remark}
\begin{definition}\cite{CONNES}\label{spectraldistance} The spectral distance between any two states is given by:
\begin{equation}
d_D(\omega_1,\omega_2)=\sup_{a\in\mathbb{A}}\big\{ \vert \omega_1(a)-\omega_2(a)\vert;\ \ell_D(a)\le1\big \},\ \forall\omega_1,\omega_2\in{\mathfrak{S}}(\mathbb{A}) \label{specdist},
\end{equation}
in which the seminorm on $\mathbb{A}$, $\ell_D:\mathbb{A}\to{\mathbb{R}}^+$, associated with the Dirac operator $D$  in Definition \ref{spectraltriple} is defined by
\begin{equation}
\ell_D:a\mapsto \ell_D(a):=||[D,\pi(a)]||,\  \forall a\in\mathbb{A}\label{lipnorm}.
\end{equation}
\end{definition}
One can check that \eqref{specdist} satisfies all the axioms for a distance except that it can reach $+\infty$ if no further conditions are imposed, i.e \eqref{specdist} defines a pseudo-distance. Note that the supremum may be searched only on the self-adjoint elements of $\mathbb{A}$. Indeed, if some element $a_0\in\mathbb{A}$ reaches the supremum so does $\frac{1}{2}(a_0+a_0^\dag)$, stemming simply from $\omega(a^\dag)=\omega(a)^\star$ for any $a\in\mathbb{A}$ and $\ell_D(\frac{1}{2}(a_0+a_0^\dag))\le1$.\par 

According to the above discussion, it is natural to interpret a {\it{unital}} spectral triple whose the spectral distance \eqref{specdist} metrizes the w*-topology on $\mathfrak{S}(\mathbb{A})$ as a compact quantum metric space with semi-norm $L(.)=\ell_D(.)$. The condition for the spectral distance \eqref{specdist} to define a metric on $\mathfrak{S}(\mathbb{A})$ is summarized in the following proposition \cite{rennie-var}
\begin{proposition}[\cite{rennie-var}]\label{prop-unital-metric}
For a unital spectral triple $(\mathbb{A},\ \pi,\ \mathcal{H},\ D)$, the spectral distance related to $D$ is a metric on $\mathfrak{S}(\mathbb{A})$ if $\pi$ is non-degenerate, i.e $\mathbb{A}\mathcal{H}=\mathcal{H}$ and the metric commutant $\mathbb{A}^\prime_D=\{a\in\mathbb{A};\ [D,\pi(a)]=0 \}$ is trivial.
\end{proposition}
Note that the second condition simply means that the metric commutant is $\mathbb{A}=\mathbb{R}\bbone$ which is nothing but the necessary condition given in \cite{Rieffel2a} for a semi-norm on $\mathbb{A}$ to define a quantum metric space. Provided Proposition \ref{prop-unital-metric} holds true, $d_D$ metrizes the w*-topology on ${\mathfrak{S}}(\mathbb{A})$ if and only if the ``Lipschitz ball'' $B({\mathfrak{X}}_D):=\{a\in\mathbb{A};\ \ell_D(a)\le1\big \}$ is norm pre-compact in  $\mathbb{A}/\mathbb{A}_D^\prime$. This is the condition for a unital spectral triple to give rise to a compact quantum metric space.\par 

The extension to the non unital case is not straightforward. A natural extension has been achieved only recently \cite{latrem2}, \cite{latrem3}. The point is that the w*-topology requirement mentioned above has to be reconsidered. In fact the space of states for a non unital (C*-)algebra is not closed under the w*-topology. For commutative locally compact metric space $X$, the Kantorovitch distance does not generally give rise to the w*-topology on $\mathfrak{S}(C_0(X))$. A first progress was made in \cite{latrem1} by exploiting the observation that any distance belonging to the family of bounded-Lipschitz distances metrizes the w*-topology on $\mathfrak{S}(C_0(X))$. Recall that the set of bounded-Lipschitz distances $(d^\alpha_\ell)_{\alpha\in\mathbb{R}^+}$ is defined by 
\begin{equation}
d^\alpha_\ell(\mu_1,\mu_2):=\sup_{f\in C_0(X)}(|\mu_1(f)-\mu_2(f)|;\ \ell(f)\le1,\ ||f||\le\alpha),
\end{equation}
for any $\mu_1,\mu_2\in{\mathfrak{S}}(C_0(X))$. A natural extension of this observation amounts to consider a non unital C*-algebra $\mathbb{A}$, assumed to be separable for technical convenience in \cite{latrem1}, together with $L$ a seminorm (defined for technical convenience on a dense subset{\footnote{with extension $L(a)=+\infty$ for any $a$ not in the domain.}} of $\mathbb{A}_{sa}$, the self-adjoint part of $\mathbb{A}$) and the following one-parameter family $(\mathfrak{B})_{r},\ \forall r\in\mathbb{R}^+$, $r\ne0$, such that 
\begin{equation}
\mathfrak{B}_r:=\{ a\in\mathbb{A}\ ;\ L(a)\le1,\ ||a||_\mathbb{A}\le r\}
\end{equation}
and to determine the conditions satisfied by $L$ so that $d_{\mathfrak{B}_r}$ defined by
\begin{equation}
d_{\mathfrak{B}_r}:=\sup\{|\omega_1(a)-\omega_2(a)|\ ;\ a\in\mathfrak{B}_r \} 
\end{equation}
metrizes the w*-topology on $\mathfrak{S}(\mathbb{A})$. An answer is synthesized by the main theorem of \cite{latrem1}
\begin{theorem}[\cite{latrem1}]\label{bounded-noncompact}
Let $\mathbb{A}$ be a non unital separable C*-algebra, $\mathfrak{B}$ a totally bounded subset of $\mathbb{A}_{sa}$ and for any two $\omega_1, \omega_2\in\mathfrak{S}(\mathbb{A})$ one defines
\begin{equation}
d_{\mathfrak{B}}(\omega_1,\omega_2):=\sup\{|\omega_1(a)-\omega_2(a)|\ ;\ a\in\mathfrak{B} \}.
\end{equation}
Then $d_{\mathfrak{B}}$ is a distance on $\mathfrak{S}(\mathbb{A})$ and the following properties are equivalent:\\
i) There exists a strictly positive $h\in\mathbb{A}$, $h\mathfrak{B}h$ is totally bounded for $||.||_\mathbb{A}$.\\
ii) $\mathfrak{B}$ is totally bounded in the weak uniform topology on $\mathbb{A}$.\\
iii) $d_{\mathfrak{B}}$ metrizes the restriction of the w*-topology $\sigma(\mathbb{A}^\star,\mathbb{A})$ to $\mathfrak{S}(\mathbb{A})$.
\end{theorem}
When $\mathfrak{B}$ is not norm bounded, the combination of Theorem \ref{bounded-noncompact} with the replacement of $d_{\mathfrak{B}}$ by $ d_{\mathfrak{B}_r}$ defined above allows one to obtain a criterion for these latter metrics to metrize the w*-topology on $\mathfrak{S}(\mathbb{A})$ which was the way followed in \cite{latrem1}. Of course, when $\mathfrak{B}$ is norm bounded with $\mathfrak{B}=\{a\in\mathbb{A}_{sa}\ ; \ L(a)\le1 \}$, $d_{\mathfrak{B}}$ is simply the Kantorovich distance related to $L$ for which the diameter of $\mathfrak{S}(\mathbb{A})$ is finite and Theorem \ref{bounded-noncompact} provides the conditions to have w*-topology on $\mathfrak{S}(\mathbb{A})$. The translation of these results into the framework of spectral triples is obvious. Note that one can show that the space of states is complete for the spectral distance \cite{latrem1}.\par

This first attempt has been extended recently in \cite{latrem2} to give rise to a complete description of a notion of quantum locally compact (metric) space. It is based on the observation made long ago in \cite{dobrush} that the Kantorovitch metric for non compact (commutative) metric spaces is in fact well behaved when restricted to subsets of the space of states with suitable behavior at infinity. There is no canonical generalization of the notion of behavior at infinity for noncommutative C*-algebra. In view of a possible relationship to the Gromov-Hausdorff distance, the proposal made in \cite{latrem2} also faced the problem of a proper extention of the notion of locality. The solution proposed in \cite{latrem2} amounts roughly to use a prefered commutative set of elements which supports a notion of locality, completed with other elements that do not commute with that prefered set ``spread" all over. This is formalized in \cite{latrem2} by a pair $(\mathbb{A},\mathfrak{M})$ which is called a ``topographic quantum space" where $\mathbb{A}$ is a non unital C*-algebra and $\mathfrak{M}$ a commutative C*-subalgebra with approximate unit, this latter approximate unit being used to deal with the notion of behavior at infinity. From this, combined with a suitable seminorm on the algebra, follows the definition of a local compact quantum space of \cite{latrem2}.

\subsection{Summary of the paper.}

Section \ref{sec2} gives a survey on recent results obtained for Moyal plane. The relevant technical material is collected in Subsection \ref{subsection21}. Subsection \ref{subsection22} deals with the standard description of Moyal plane through a spectral triple that is currently used in the mathematical physics. Within this description, an explicit distance formula can be computed between a particular class of pure states, as first shown in \cite{Wal1} and \cite{Wal2}. This is summarized by Theorem \ref{th0}. Variations of the above spectral triple also appeared in mathematical physics. The corresponding families of spectral triples are considered in Subsection \ref{subsection23} from the viewpoint of the related spectral distances that can be shown \cite{Wallet:2011aa} to be simply related to the spectral distance for the above ``standard'' Moyal plane, as summarized by Theorem \ref{th1}. In Subsection \ref{subsection24}, metric properties of the above noncommutative spaces are presented, emphasizing the fact that Moyal plane and its variations considered here provide examples of quantum locally compact spaces as defined in \cite{latrem2}.

Section \ref{sec3} presents a new result concerning the possibility of causal relations between coherent states on Moyal plane. Subsection \ref{subsection31} gives a short review of the construction of spectral triples with a Lorentzian signature and the notion of causality on them as introduced in \cite{CQG2013}, a notion of causality which corresponds to the usual one when the algebra is commutative. The technical ingredients for the adaptation of the spectral triple based on Moyal plane to Lorentzian signature are collected in Subsection \ref{subsection32}. In Subsection \ref{subsection33}, coherent states on Moyal plane are defined and the main result (Theorem \ref{mainthmcausality}) giving the exact causal structure between coherent states is presented and proved.

\section{Triples on Moyal plane as quantum locally compact metric spaces.}\label{sec2}

\subsection{Basics.}\label{subsection21}
Let $\mathcal{S}$ be the space of $\mathbb{C}$-valued Schwartz functions on $\mathbb{R}^2$. We define
\begin{equation}
\Theta_{\mu\nu}:=\theta\begin{pmatrix} 0&1 \\ -1& 0 \end{pmatrix},
\end{equation}
where $\theta\in\mathbb{R}$, $\theta>0$.
\begin{proposition}\label{moyalproduct}\cite{Gracia-Bondia:1987kw, Varilly:1988jk}
The Moyal $\star$-product is defined by: $\star:\mathcal{S}\times\mathcal{S}\to\mathcal{S}$ 
\begin{equation}
\label{moyal}
(f\star g)(x):=\frac{1}{(\pi\theta)^2}\int d^2y\,d^2z\ f(x+y)g(x+z)e^{-2i\,y^\mu\,\Theta^{-1}_{\mu\nu}z^\nu},\ 
\end{equation}
for all $f, g$ in $\mathcal{S}$. The complex conjugation is an involution for $\star$. The Lebesgue integral is a faithful trace. One has $
\int d^2x\ (f\star g)(x)=\int d^2x\ (g\star f)(x)=\int d^2x\ f(x)g(x)$, $\partial_\mu(f\star g)=\partial_\mu f\star g+f\star\partial_\mu g$, $\forall f,g\in\mathcal{S}$.
\end{proposition}
We set $\mathbb{A}:=(\mathcal{S},\star)$. Various properties that will be needed for the sequel are recalled below.
\begin{proposition}[\cite{GAYRAL2004}]\label{precstar}
$\mathbb{A}$ is a non unital Fr\'echet pre-C* algebra when equipped with usual complex conjugation and Fr\'echet-type seminorms defined by $\rho_{\alpha,\beta}(a)=\sup_{x\in\mathbb{R}^2}|x_1^{\alpha_1}x_2^{\alpha_2}\partial_1^{\beta_1}\partial_2^{\beta_2}a|$, $\forall a\in\mathbb{A}$, where $\alpha=(\alpha_1,\alpha_2)$ and $\beta=(\beta_1,\beta_2)$.
\end{proposition}
We will need to use some explicit expressions for the so-called matrix basis, a Hilbertian basis for $\mathcal{S}$, defined by the family of Wigner transition eigenfunctions of the harmonic oscillator denoted by $\{f_{mn}\}_{m,n\in\mathbb{N}}\subset\mathcal{S}$. Introduce the set of local complex coordinates $z=\frac{1}{\sqrt{2}}(x_1+ix_2)$, $\bar{z}=\frac{1}{\sqrt{2}}(x_1-ix_2)$.
\begin{proposition}[\cite{Gracia-Bondia:1987kw, Varilly:1988jk}]\label{matrixbase-prop}
The elements of the matrix basis satisfy:
\begin{equation}
f_{mn}=\frac{1}{\sqrt{\theta^{m+n}m!n!}}\bar{z}^{\star m}\star f_{00}\star z^{\star n},\ f_{00}=2e^{-\frac{2}{\theta}H},\ H=\bar{z}z,
\end{equation}
\begin{equation}
f_{mn}\star f_{kl}=\delta_{nk}f_{ml},\ f_{mn}^\dag=f_{nm},\ \langle f_{mn},f_{kl} \rangle_{L^2}=2\pi\theta\delta_{mk}\delta_{nl}.\label{matrix-base-basic}
\end{equation}
There is a Fr\'echet algebra isomorphism between $\mathbb{A}$ and $\mathcal{M}_\theta$, the Fr\'echet *-algebra of doubly indexed sequences $\{\Phi_{p,k}\}_{p,k\in\mathbb{N}}$ equipped with the canonical matrix product and the family of seminorms
\begin{equation}
\rho^2_n(\Phi)=\sum_{p,k}\theta(p+\frac{1}{2})^n(k+\frac{1}{2})^n|\Phi_{pk}|^2<\infty,\ \forall n\in\mathbb{N}.\nonumber
\end{equation}
 The isomorphism is defined by $\Phi_{mn}\mapsto \sum_{m,n}\Phi_{mn}f_{mn}\in\mathcal{S}$ and inverse $\Phi\in\mathcal{S}\mapsto\frac{1}{2\pi\theta}\langle \Phi,f_{mn}\rangle_{L^2}$.
\end{proposition}
There is a natural approximate unit for $\mathbb{A}$ that we will explicitly need in a while.
\begin{proposition}\label{approx-unit}
Define for any $N\in\mathbb{N}$, $e_N(x):=\sum_{k\le N}f_{kk}(x)$. Then, $e_N(x)$ is an approximate unit for $\mathbb{A}$.
\end{proposition}
\begin{proof}
For any $a\in\mathbb{A}$, $a(x)=\sum_{m,n}a_{mn}f_{mn}(x)$, one computes the quantities $(e_N\star a)(x)=\sum_{m\le N}\sum_{n}a_{mn}f_{mn}(x)$ and $(a\star e_N)(x)=\sum_{m}\sum_{n\le N}a_{mn}f_{mn}(x) $. Then, one has  $\lim_{N\to\infty}(e_N\star a)=\lim_{N\to\infty}(a\star e_N)=a $.
\end{proof}
\begin{remark}
The $\star$-product can be extended to larger subspaces of $\mathcal{S}^\prime$, the space of tempered distributions. One introduces the family of spaces $\{\mathcal{G}_{s,t}\}_{s,t\in\mathbb{R}}$
\begin{equation}
\mathcal{G}_{s,t}:=\{ a=\sum_{m,n\in{\mathbb{N}}}a_{mn}f_{mn}\in\mathcal{S}^\prime;\ ||a ||_{s,t}^2<\infty\}, \label{gst}
\end{equation}
with $||a ||_{s,t}^2=\sum_{m,n}\theta^{s+t}\big(m+\tfrac{1}{2}\big)^s\big(n+\tfrac{1}{2})^t|a_{mn}|^2$ by completing $\mathcal{S}$ for the norm $||.||_{s,t}$. For any $a=\sum_{m,n}a_{mn}f_{mn}\in\mathcal{G}_{s,t}$ and any $b\in\mathcal{G}_{q,r}$, $b=\sum_{m,n}b_{mn}f_{mn}$, with $t+q\ge0$, the sequences $c_{mn}=\sum_{p}a_{mp}b_{pn}$, $\forall\, m,n\in\mathbb{N}$ define the functions $c=\sum_{m,n}c_{mn}f_{mn}$, $c\in\mathcal{G}_{s,r}$, as $||a\star b||_{s,r}\le||a||_{s,t}||b||_{q,r}$, $t+q\ge0$ and $|| a||_{u,v}\le ||a ||_{s,t}$ if $u\le s$ and $v\le t$. For more details, see e.g.~\cite{Gracia-Bondia:1987kw, Varilly:1988jk}. In particular, $\mathcal{G}_{0,0}=L^2(\mathbb{R}^2)$ and the dense and continuous inclusion $\mathcal{S}\subset\mathcal{G}_{s,t}\subset\mathcal{S}^\prime$ holds true for any $s,t\in\mathbb{R}$.\par 
\end{remark}

A slightly more explicit characterization of the algebra $\mathbb{A}$, that will be useful below to determine completely the space of pure states of $\mathbb{A}$ (see the remark \ref{precstar-states} in the Introduction), can be done by exploiting the Fr\'echet isomorphism of Proposition \ref{matrixbase-prop}. Let $(e_n)_{n\in\mathbb{N}}$ be the canonical basis of $\ell^2(\mathbb{N})$, $\mathcal{B}(\ell^2(\mathbb{N}))$ the C*-algebra of bounded operators on $\ell^2(\mathbb{N})$, $\mathcal{K}(\ell^2(\mathbb{N}))$ the C*-subalgebra of compact operators. We denote by $\eta:\ell^2(\mathbb{N})\otimes\ell^2(\mathbb{N})\to\mathcal{B}(\ell^2(\mathbb{N})) $ the natural representation of the elements of $\mathcal{M}_\theta$ on $\ell^2(\mathbb{N})$, i.e simply defined by the product of a matrix by a column vector. For any $\Phi\in\mathcal{M}_\theta$, $\Phi=\sum_{m,n}\Phi_{mn}e_m\otimes e_n$, define $\eta(e_m\otimes e_n)=e_m\otimes e^*_n,\ \forall m,n\in\mathbb{N}$ with $e^*_n(e_p)=\delta_{np}$. \\
Then, from the definition of the seminorm $\rho_0(\Phi)$ (see in Proposition \ref{matrixbase-prop}), one infers that $\eta(\mathcal{M}_\theta)\subset {\mathcal{L}}^2(\ell^2(\mathbb{N}))$ where ${\mathcal{L}}^2(\ell^2(\mathbb{N}))$ is the set of Hilbert-Schmidt operators on $\ell^2(\mathbb{N})$. But ${\mathcal{L}}^2(\ell^2(\mathbb{N}))\subset \mathcal{K}(\ell^2(\mathbb{N}))$. Hence $\overline{\eta(\mathcal{M}_\theta)}\subset\mathcal{K}(\ell^2(\mathbb{N}))$ where $\overline{\eta(\mathcal{M}_\theta)}$ is the closure of $\eta(\mathcal{M}_\theta)$ in the operator norm.\\
Now recall that $\eta$ is faithful. Thus, one has the isomorphism $\eta(\mathcal{M}_\theta)\simeq\mathcal{M}_\theta$. Hence $\overline{\mathcal{M}_\theta}\subset\mathcal{K}(\ell^2(\mathbb{N}))$. On the other hand, $\overline{\mathcal{M}_\theta}\supset\mathcal{K}(\ell^2(\mathbb{N}))$ since one has $\mathcal{M}_\theta\supset\mathbb{M}_{\infty}(\mathbb{C})$ with  $\mathbb{M}_{\infty}(\mathbb{C}):=\bigcup_{n=1}^\infty\mathbb{M}_n(\mathbb{C})$. But the closure of $\mathbb{M}_{\infty}(\mathbb{C})$ is $\mathcal{K}(\ell^2(\mathbb{N}))$ in view of \cite{viperenoire} (II.8.2.2). Thus $\overline{\mathcal{M}_\theta}=\mathcal{K}(\ell^2(\mathbb{N}))$.\\ Finally, the map $U:L^2(
\mathbb{R}^2)\to\ell^2(\mathbb{N})\otimes\ell^2(\mathbb{N})$ defined by $U:f_{mn}\mapsto {\sqrt{2\pi\theta}}e_m\otimes e_n$ for any $m,n\in\mathbb{N}$ is an isometry. Therefore, one concludes that
\begin{equation}
\overline{\mathbb{A}}\simeq\mathcal{K}(\ell^2(\mathbb{N})).
\end{equation}

We are in position to determine the space of pure states of $\mathbb{A}$.
\begin{proposition} (\cite{Wal1}, \cite{Wal2})\label{purestatesmoyal}
The pure states of $\mathbb{A}$ are vector states defined by $\omega_\psi:\mathbb{A}\to\mathbb{C}$ where $\psi$ is a unit vector of $\ell^2(\mathbb{N})$, $\psi=\sqrt{2\pi\theta}\sum_m\psi_me_m$ with $2\pi\theta\sum_m|\psi_m|^2=1$ and one has
\begin{equation}
\omega_\psi(a)=2\pi\theta\sum_{m,n}\psi^*_ma_{mn}\psi_n, \forall a\in\mathbb{A}.\label{omegapsi}
\end{equation}
\end{proposition}
\begin{proof}
All the pure states of $\mathcal{K}(\ell^2(\mathbb{N}))$ are vector states for the irreducible representation on $\ell^2(\mathbb{N})$ by Corollary 10.4.4 of \cite{Kadison}. Then, observe that the pure states of $\mathbb{A}$ are pure states of $\bar{\mathbb{A}}$ while Equation \eqref{omegapsi} follows from simple calculation. This completes the proof.
\end{proof}

\subsection{Standard Moyal spectral triple and distance formula.}\label{subsection22}

We now introduce the standard Moyal spectral triple together with some necessary technical ingredients. Set $\partial:={{1}\over{{\sqrt{2}}}}(\partial_1-i\partial_2)$, $\delbar:={{1}\over{{\sqrt{2}}}}(\partial_1+i\partial_2)$. The self-adjoint Dirac operator on  $\mathbb{R}^2$ is
\begin{equation}
D_0:=-i\sigma^\mu\del_\mu = -i{\sqrt{2}}\begin{pmatrix} 0&\delbar \\ \del& 0 \end{pmatrix},\ \sigma^1=\begin{pmatrix} 0&1 \\ 1& 0 \end{pmatrix},\ 
\sigma^2=\begin{pmatrix} 0&i \\ -i& 0 \end{pmatrix}\label{cliff2},
\end{equation}
with $\text{Dom}(D_0)=\mathcal{D}_{L^2}\otimes \mathbb{C}^2$. Here $\mathcal{D}_{L^2}$ is the set of smooth functions of $L^2(\mathbb{R}^2)$ with all their derivatives in $L^2(\mathbb{R}^2)$. One has $\mathbb{A}\subset\mathcal{D}_{L^2}$. The $\sigma^\mu$'s ($\mu=1,2$) span an irreducible representation of the complex Clifford algebra $\mathbb{C}l_{\mathbb{C}}(2)$, $\sigma^\mu\sigma^\nu+\sigma^\nu\sigma^\mu=2\delta^{\mu\nu}$ and we set $\sigma^3:=i\sigma^1\sigma^2$. In this section $\mathcal{H}_0= L^2(\mathbb{R}^2)\otimes \mathbb{C}^2$, i.e it is the Hilbert space of square integrable sections of the trivial spinor bundle $\mathbb{R}^2\times\mathbb{C}^2$. The corresponding inner product is
\begin{equation}
\langle \psi,\phi\rangle=\int d^2x(\psi_1^*\phi_1+\psi_2^*\phi_2)
\end{equation}
for any $\psi = \binom{ \psi_1 }{ \psi_2 }
,\;\phi = \binom{ \phi_1 }{ \phi_2 }\; \in\mathcal{H}_0$. We define $L(a)\in\mathcal{B}(L^2(\mathbb{R}^2))$ by 
\begin{equation}
L(a)\psi:=a\star\psi, \nonumber
\end{equation}
for any $a\in\mathbb{A}$ and any $\psi\in L^2(\mathbb{R}^2)$. One has $L(a)^\dag=L(a^\star)$. We denote by $\pi_0:\mathbb{A}\to\mathcal{B}(\mathcal{H}_0)$, the faithful left regular representation of $\mathbb{A}$ on $\mathcal{B}(\mathcal{H}_0)$:
\begin{equation}
\pi_0(a):=L(a)\otimes\bbone_2,\ \pi_0(a)\psi=\binom{a\star\psi_1}{a\star\psi_2},\label{regulrep}
\end{equation}
for any $a\in\mathbb{A}$ and any $\psi=\binom{\psi_1}{\psi_2}\in\mathcal{H}_0$. 
The following useful property for the Lipschitz seminorm on $\mathbb{A}$ for $D_0$ holds true.
\begin{proposition}\label{lzero}
We set $\ell_{D_0}(a):=||[D_0,\pi_0(a)] ||$, for any $a\in\mathbb{A}$. Then, one has
$\ell_{D_0}(a)={\sqrt{2}}\max(||{{L}}(\partial a)||,\ ||{{L}}(\delbar a)||), \forall a\in\mathbb{A}$.
\end{proposition}
\begin{proof}
A standard computation yields $[\partial_\mu,L(a)]=L(\partial_\mu(a))$. Thus $[D_0,\pi(a)]=-iL(\partial_\mu  a)\otimes\sigma^\mu$ for any $a\in\mathbb{A}$. Then, the application of the general property $||T||^2=||TT^\dagger||$ to the operator $T=[D_0,\pi(a)]$ leads after a simple computation to the result.
\end{proof}
The (by now) standard Moyal spectral triple is described by the following proposition whose proof is reproduced below as it illustrates useful technical properties of the Moyal product.
\begin{proposition}\label{X0}\cite{GAYRAL2004}
The data ${\mathfrak{X}}_{D_0}=(\mathbb{A},\ \pi_0,\ \mathcal{H}_0=L^2(\mathbb{R}^2)\otimes \mathbb{C}^2,\ D_0)$ where $\pi_0:\mathbb{A}\to\mathcal{B}(\mathcal{H}_0)$ is defined in \ref{regulrep} and $D_0$ is defined in \eqref{cliff2} is a spectral triple as in Definition\ref{spectraltriple}.
\end{proposition}
\begin{proof}
Axiom i and Axiom ii of Definition \ref{spectraltriple} are mere consequences of basic properties of Moyal product, algebras into play and the Dirac operator $D_0$. Besides, it can be easily verified that $||a\star b||_2\le{{1}\over{{\sqrt{2\pi\theta}}}}||a ||_2||b||_2$, $\forall a,b\in L^2(\mathbb{R}^2)$. This implies that $L(a)$ is a bounded operator. Thus, $\pi_0(a)\in\mathcal{B}(\mathcal{H}_0)$. Then, Proposition \ref{lzero} implies $[D_0,\pi_0(a)]\in\mathcal{B}(\mathcal{H}_0)$ for any $a\in\mathbb{A}$. Hence, Axiom iii is verified.\\
Next, one verifies that $D_0^2=-\partial_\mu\partial^\mu\otimes\bbone_2:=-\partial^2\otimes\bbone_2$ which combined with \eqref{regulrep} yields
\begin{equation}
\pi_0(a){{1}\over{D_0^2+1}}=L(a){{1}\over{-\partial^2+1}}\otimes\bbone_2\label{diagresolv},
\end{equation}
which acts diagonally on $\mathcal{H}_0=L^2(\mathbb{R}^2)\otimes\mathbb{C}^2$, so that it is sufficient to show that $L(a){{1}\over{-\partial^2+1}}$ is a compact operator on $L^2(\mathbb{R}^2)$ for any $a\in\mathbb{A}$. Simply compute the integral kernel for $L(a)$ given by
\begin{equation}
K_{L(a)}(x,y)=\int d^2z\ a(x+z)e^{2i\,z^\mu\,\Theta^{-1}_{\mu\nu}(x^\nu-y^\nu)},\ \forall a\in\mathbb{A}\label{kla}
\end{equation}
from which follows ($C$ is a real constant)
\begin{equation}
K_{L(a)(-\partial^2+1)^{-1}}(x,y)=C\int d^2p\ {{a(x+{{1}\over{2}}\Theta p)}\over{p^2+1}}e^{ip(x-y)},\ \forall a\in\mathbb{A}\label{kernelresolvante1}.
\end{equation}
Then, consider the integral $I:=\int d^2xd^2y|K_{L(a)(-\partial^2+1)^{-1}}(x,y)|^2$. One has
\begin{align}
I=C^2\int d^2xd^2yd^2p_1d^2p_2\ {{a^*(x+{{1}\over{2}}\Theta p_1)}\over{p_1^2+1}}
{{a(x+{{1}\over{2}}\Theta p_2)}\over{p_2^2+1}}e^{-ip_1(x-y)}e^{ip_2(x-y)} \nonumber\\
=C^2\int d^2xd^2p\ {{a^*(x+{{1}\over{2}}\Theta p)}\over{p^2+1}}
{{a(x+{{1}\over{2}}\Theta p)}\over{p^2+1}}\nonumber\\
=C^2\int d^2xd^2p\ |a(x+{{1}\over{2}}\Theta p)|^2({{1}\over{p^2+1}})^2
\end{align}
Thus
\begin{equation}
I=(C^\prime)^2\ ||a||^2_2\ ||{{1}\over{p^2+1}}||^2_2<+\infty,\ \forall a\in\mathbb{A},\label{HS1}
\end{equation}
where the last equality stems from a change of variable. Relation \eqref{HS1} implies that the operator $L(a)(-\partial^2+1)^{-1}$ is a Hilbert-Schmidt operator on $L^2(\mathbb{R}^2)$, hence compact. This implies that Axiom iv is satisfied.
\end{proof}

It is possible to actually compute an explicit distance formula between a particular class of pures states, which has been done in \cite{Wal1}, \cite{Wal2}. These states are defined by the vectors of the canonical basis of $\ell^2(\mathbb{N})$. Combining Proposition \ref{purestatesmoyal} with $\psi=e_n$, for any $n\in\mathbb{N}$, these states are given by
\begin{equation}
\omega_{e_n}(a):=\omega_n(a)=a_{nn},\ \forall n\in\mathbb{N},\label{diagon-states}
\end{equation}
which thus define the diagonal elements of the "matrix" $a_{mn}$ representing $a\in\mathbb{A}$. We recall that, according to Definition \ref{spectraldistance}, the spectral distance for $D_0$ is
\begin{equation}
d_{D_0}(\omega_m,\omega_n)=\sup_{a\in\mathbb{A}}\{|\omega_m(a)-\omega_n(a)|\ 
;\ \ell_{D_0}(a)\le1 \},
\end{equation}
for any $\omega_m,\ \omega_n$ given by \eqref{diagon-states}. According to the material presented in the introduction, we define the Lipschitz ball for $\ell_{D_0}$ as
\begin{equation} 
{\mathfrak{B}_{\ell_{D_0}}:=\{a\in\mathbb{A}\ ;\ \ell_{D_0}(a)\le1} \}\nonumber.
\end{equation}
\begin{theorem}(\cite{Wal1}, \cite{Wal2})\label{th0}
The spectral distance between any two pure states $\omega_m$, $\omega_n$ is
\begin{equation}
d_{D_0}(\omega_m,\omega_n)=\sqrt{\frac{\theta}{2}}\sum_{k=n+1}^m\frac{1}{\sqrt{k}}, \forall m,n\in\mathbb{N},\ n<m.
\end{equation}
\end{theorem}
\begin{proof}
Assume that $a\in\mathfrak{B}_{\ell_{D_0}}$ is such that $\ell_{D_0}(a)=\sqrt{2}||L(\partial a) ||\le1$. The proof for $L(\bar{\partial}a)$ is similar. Then, Cauchy-Schwarz inequality implies 
\begin{equation}
|\langle \phi_1,L(\partial a)\phi_2 \rangle_{L^2} |\le\frac{1}{\sqrt{2}}\label{cauchy-addit}
\end{equation}
for any unit vectors $\phi_1$, $\phi_2$ in $L^2(\mathbb{R}^2)$. Choose $\phi_1=f_{N+1,0}$ and $\phi_2=f_{N,0}$. Upon using $\partial f_{mn}=\sqrt{\frac{n}{\theta}}f_{m,n-1}-\sqrt{\frac{m+1}{\theta}}f_{m+1,n}$ for any $m,n\in\mathbb{N}$, which results from a simple calculation, one readily obtains
$\sqrt{\frac{N+1}{\theta}}|a_{N+1,N+1}-a_{NN} |\le\frac{1}{\sqrt{2}}$. Thus, 
by \eqref{diagon-states}, we can write
\begin{equation}
d_{D_0}(\omega_{N+1},\omega_N)\le\sqrt{\frac{\theta}{2(N+1)}},\label{bound1}
\end{equation}
and by the triangular inequality we also have
\begin{equation}
d_{D_0}(\omega_M,\omega_N)\le\sqrt{\frac{\theta}{2}}\sum_{k=N+1}^M\frac{1}{\sqrt{k}},\ \forall M,N\in\mathbb{N},\ N<M.\label{bound2}
\end{equation}
Now, the bound in the RHS of \eqref{bound1} is saturated by the element of $\mathbb{A}$ given by $\hat{a}_1=\sqrt{\frac{\theta}{2(N+1)}}f_{NN}$. Define for convenience $\partial a=\sum_{m,n}\alpha_{mn}f_{mn}$ for any $a\in\mathbb{A}$. On one hand, a mere calculation of the derivative $\partial \hat{a}_1$ yields $\alpha_{N+1,N}=\frac{1}{\sqrt{2}}$ while the other coefficients vanish. On the other hand, one can write for any unit vector $\psi\in L^2(\mathbb{R}^2)$, $\psi=\sum_{m,n}\psi_{mn}f_{mn}$, $||\partial \hat{a}_1\star\psi ||_2^2=2\pi\theta\sum_{k,q}|\alpha_{k+1,k}\psi_{kq}|^2\le\pi\theta\sum_{k,q}|\psi_{kq}|^2=\frac{1}{2}$. Hence $\hat{a}_1\in\mathfrak{B}_{\ell_{D_0}}$ and 
\begin{equation}
d_{D_0}(\omega_{N+1},\omega_N)=\sqrt{\frac{\theta}{2(N+1)}}.
\end{equation}
Now, define 
\begin{equation}
\hat{a}(m_0)=\sum_{p,q}a_{pq}(m_0)f_{pq},\  
a_{pq}(M)=
\delta_{pq}\sqrt{\frac{\theta}{2}}\sum_{k=p}^{m_0}\frac{1}{\sqrt{k+1}}, \forall m_0\in\mathbb{N} \label{positif-rad}
\end{equation}
Observe that $\hat{a}(m_0)$ is a positive element of $\mathbb{A}$ so that 
\begin{equation}
\ell_{D_0}(\hat{a}(m_0))=\sqrt{2}||L(\partial \hat{a}(m_0)) ||. 
\end{equation}
By computing the derivative of $\hat{A}(m_0)$, one finds that the only non-vanishing coefficients of the corresponding expansion are $\alpha_{p+1,p}=-\frac{1}{\sqrt{2}}$ with $0\le p\le m_0$. Again $||L(\partial\hat{a}(m_0))||^2\le\frac{1}{2} $. Thus $\hat{a}(m_0)\in\mathfrak{B}_{\ell_{D_0}}$. Finally, to show that any $\hat{a}(m_0)$ saturates the bound of the RHS of \eqref{bound2} when $n<M\le m_0$, compute
\begin{eqnarray}
|\omega_M(\hat{a}(m_0))-\omega_N(\hat{a}(m_0)) |&=&\sqrt{\frac{\theta}{2}}|\sum_{k=M}^{m_0}\frac{1}{\sqrt{k+1}}-\sum_{k=N}^{m_0}\frac{1}{\sqrt{k+1}} |\nonumber\\
&=&\sqrt{{\frac{\theta}{2}}}\sum_{k=N+1}^M\frac{1}{\sqrt{k}}\label{triang-equal}.
\end{eqnarray}
This terminates the proof.
\end{proof}
\begin{remark}
An additional information on the derivatives of the elements of $\mathfrak{B}_{\ell_{D_0}}$ can be obtained from \eqref{cauchy-addit}. For any $a\in\mathfrak{B}_{\ell_{D_0}}$ and setting $\partial a=\sum_{m,n}\alpha_{mn}f_{mn}$ as above and using again \eqref{cauchy-addit} with 
$\phi_1=f_{p0}$ and $\phi_2=f_{q0}$ yields $|\alpha_{pq}|\le\frac{1}{\sqrt{2}}$, for any $p,q\in\mathbb{N}$. Similarly, $\bar{\partial }a=\sum_{m,n}\beta_{mn}f_{mn}$ leads to $|\beta_{pq}|\le\frac{1}{\sqrt{2}}$ for any $p,q\in\mathbb{N}$. For positive elements $a_+$, one has the equivalence
\begin{equation}
a_+\in\mathfrak{B}_{\ell_{D_0}}\iff|\alpha_{pq}|\le\frac{1}{\sqrt{2}}\ \text{and}\ |\beta_{pq}|\le\frac{1}{\sqrt{2}}.
\end{equation}
This can be verified by noticing that the only non-vanishing coefficients $\alpha_{mn}$ are of the form $\alpha_{m+1,m}$. Then, assuming $|\alpha_{m+1,m}|\le\frac{1}{\sqrt{2}}$ and merely computing $||\partial a\star\psi||^2$ for any unit vector $\psi$ leads immediately to the result.\\
Besides, as it can be realized from \eqref{triang-equal}, one has $d_{D_0}(\omega_m,\omega_n)=d_{D_0}(\omega_m,\omega_p)+d_{D_0}(\omega_p,\omega_n)$ for $n\le p\le m$, i.e the spectral distance satisfies a triangular equality among these pures states.
\end{remark}

\subsection{Homothetic metrics on Moyal planes.}\label{subsection23}

We now introduce another family of Moyal spectral triples that appeared in the mathematical physics literature in the general context of field theories and gauge theories built on noncommutative spaces (for a complete review including essential aspects of noncommutative differential geometry underlying these field theories together with a list of essential references see \cite{mdv-gauge}) including the case of Moyal spaces \cite{WAL3}, \cite{WAL4}, \cite{WAL5}. This triple advocated rather recently (see \cite{wulk-finite} and related references therein) occurred in attempts to extend desirable perturbative properties (namely renormalisability) showing up in a certain class of scalar field theories on Moyal space \cite{GROSSWULK0}, \cite{GROSSEWULK}, \cite{gourwal} to the more difficult situation of gauge theories on Moyal spaces \cite{WAL6}, \cite{GROSSEWOHL}, \cite{walreview}. The use of spectral action principle leads to a gauge-invariant action with interesting properties but whose quantum properties are difficult \cite{vacuum} to study and are still under investigations. For technical reasons, renormalizability is easily obtained (at least in the scalar case) provided the spectrum of the square of the Dirac operator is a harmonic oscillator spectrum. Such a modification of the Dirac operator is standard in physics and is briefly given below together with some additional technical material. By equipping the minimal set of data for the triple by suitable grading and real structure, one arrives at a triple in the spirit of \cite{GAYRAL2004} with however a metric dimension half the KO-dimension \cite{wulk-finite}. The metric properties of this triple have been investigated in \cite{Wallet:2011aa} with the main result that the spectral distance stemming from this triple and $d_{D_0}$, the spectral distance related to the triple defined in Proposition \ref{X0} are homothetic{\footnote{The word ``homothetic" qualifying the metrics that is used here in a sense borrowed from the physical cosmology literature, e.g bigravity theories. In the present paper, it simply means ``proportional".}} to each other. This property is also valid for another type of spectral triple based on a ``Landau-type" Dirac operator as shown in \cite{Wallet:2011aa} which however will not be considered in the present paper. \par

Let $(\gamma^\mu,\gamma^{\mu+2})_{\mu=1,2}$ be hermitian elements of $\mathbb{M}_4(\mathbb{C})$ spanning an irreducible representation of the complex Clifford algebra $\mathbb{C}l_{\mathbb{C}}(4)$:
\begin{equation}
\{\gamma^\mu,\gamma^\nu\}=2\delta^{\mu\nu},\ \{\gamma^{\mu+2},\gamma^{\nu+2}\}=2\delta^{\mu\nu},\ \{\gamma^\mu,\gamma^{\nu+2}\}=0,\ \mu,\nu\in\{1,2\}\label{cliff4}.
\end{equation}
Consider the 1-parameter family of unbounded self adjoint Dirac operators indexed by $\Omega\in ]0,1]$ (to make contact with the physics literature):
\begin{equation}
D_\Omega:=\gamma^\mu(-i\partial_\mu)-\Omega\gamma^{\mu+2}m({\tilde{x}}_\mu), \label{Domega}
\end{equation}
where $m(\tilde{x}_\mu)a:=\tilde{x}_\mu a$ for any $a\in\mathbb{A}$ and $\tilde{x}_{\mu}:=2\Theta^{-1}_{\mu\nu}x^\nu$. The chosen domain of $D_\Omega$ is $\text{Dom}(D_\Omega)=\mathcal{S}\otimes\mathbb{C}^4$, dense in $\mathcal{H}=\mathcal{H}_0\otimes \mathbb{C}^2$. As (faithful) representation $\pi:\mathbb{A}\to\mathcal{B}(\mathcal{H})$, we take: 
\begin{equation}
\pi(a):=L(a)\otimes\bbone_4,\ \forall a\in\mathbb{A}\label{pi4}.
\end{equation}

Useful algebraic relations for $D_\Omega$ are collected in the following proposition.
\begin{proposition} \label{algebrorelat}(\cite{Wallet:2011aa})
The following properties hold true. 
\begin{equation}
D^2_\Omega=(-\partial^2+\Omega^2\tilde{x}^2)\bbone_4+2i\Omega\gamma^\mu\gamma^{\nu+2}\Theta^{-1}_{\nu\mu},\label{hamiltonian}
\end{equation}
\begin{equation}
[D_\Omega,\pi(a)]=-iL(\partial_\mu a)\otimes(\gamma^\mu+\Omega\gamma^{\mu+2}),\ \forall a\in\mathbb{A} \label{da}.
\end{equation}
\end{proposition}
\begin{proof}
See \cite{Wallet:2011aa}, Proposition 2.5.
\end{proof}

For convenience, we use explicit representations for the $\gamma$ matrices:  
\begin{equation}
\gamma^\mu:=\Gamma^1\otimes\sigma^\mu,\ \gamma^{\mu+2}:=\sigma^\mu\otimes\Gamma^2,\ \mu=1,2, \label{232}
\end{equation}
where $\Gamma^1$, $\Gamma^2$ are hermitian elements of $\mathbb{M}_2(\mathbb{C})$ (see below) and the two families of self-adjoint Dirac operators:
\begin{equation}
D_{1}=(\bbone_2\otimes\sigma^\mu)(-i\partial_\mu)-\Omega(\sigma^\mu\otimes\sigma^3)m({\tilde{x}}_\mu)\label{Domega1}
\end{equation}
\begin{equation}
D_{2}:=(\sigma^3\otimes\sigma^\mu)(-i\partial_\mu)-\Omega(\sigma^\mu\otimes\bbone_2)m(\tilde{x}_\mu)\label{Domega2}
\end{equation}
from which $\Gamma_1$ and $\Gamma_2$ can be read off by comparison with \eqref{Domega} and \eqref{232}. $D_{1}$ and $D_{2}$ satisfy a relation similar to Equation \eqref{hamiltonian} which is given by  $D^2_{1}=D^2_{2}=(-\partial^2+\Omega^2\tilde{x}^2)\bbone_4-{{2\Omega}\over{\theta}}(\sigma^\mu\otimes\sigma^\mu)$
while \eqref{da} still holds.\par
\begin{theorem}\label{th1}(\cite{Wallet:2011aa})
For any $\Omega\in]0,1]$ and any $k=1,2$ the data $\mathfrak{X}(k):=(\mathbb{A},\ \pi,\ \mathcal{H},\ D_k)$ where $D_k$ given by \eqref{Domega1}-\eqref{Domega2}, $\pi:\mathbb{A}\to\mathcal{B}(\mathcal{H})$ given by \eqref{pi4} and $\mathcal{H}=\mathcal{H}_0\otimes \mathbb{C}^2$ are spectral triples with spectral distances $d_{D_k}$ homothetic to $d_{D_0}$. One has
\begin{equation}
d_{D_k}(\omega_1,\omega_2)=(1+\Omega^2)^{-{{1}\over{2}}}d_{D_0}(\omega_1,\omega_2), \ \forall k=1,2,\ \forall\omega_1,\omega_2\in\mathfrak{S}(\mathbb{A})\label{distanceD}.
\end{equation}
\end{theorem}
\begin{proof}
Axiom i and Axiom ii of Definition \ref{spectraltriple} are obviously verified. Consider $k=1$. The analysis is similar for $k=2$. To verify Axiom iv, one notices from the relation $D^2_{1}=(-\partial^2+\Omega^2\tilde{x}^2)\bbone_4-{{2\Omega}\over{\theta}}$ that the spectrum of $D_1^2$ is the one of the harmonic oscillator. Hence $(D_1^2+1)^{-1}$ is (already) compact and iv) is satisfied.\\
Consider now $a\in\mathbb{A}$, $\ell_{D_1}(a)=||[D_1,\pi(a)]||=||-iL(\partial_\mu a)\otimes(\bbone_2\otimes\sigma^\mu+\Omega\sigma^\mu\otimes\sigma^3) ||$, for any $a\in\mathbb{A}$. Write explicitly:
\begin{equation}
[D_1,\pi(a)]=-i{\sqrt{2}}\begin{pmatrix} 0&L(\delbar a)&\Omega L(\delbar a)&0\\ 
L(\partial a)&0&0&-\Omega L(\delbar a)\\
\Omega L(\partial a) &0&0&L(\delbar a)\\
0&-\Omega L(\partial a)&L(\partial a)&0
             \end{pmatrix}, \forall a\in\mathbb{A}\label{matd_1a}.
\end{equation}
Set $L:=L(\partial a)$, ${\bar{L}}:=L(\delbar a)$. It follows that
\begin{align}
[D_1,\pi(a)]^*[D_1,\pi(a)]\phantom{xxxxxxxxxxxxxxxxxxxxxxxxxxxxxxx}\nonumber\\
=2\begin{pmatrix}
(1+\Omega^2)L^*L&0&0&0\\
0&{\bar{L}}^*{\bar{L}}+\Omega^2L^*L&\Omega({\bar{L}}^*{\bar{L}}-L^*L)&0\\
0&\Omega({\bar{L}}^*{\bar{L}}-L^*L)&L^*L+\Omega^2{\bar{L}}^*{\bar{L}}&0\\
0&0&0&(1+\Omega^2){\bar{L}}^*{\bar{L}}
\end{pmatrix}.
\end{align}
This can be expressed as
\begin{equation}
[D_1,\pi(a)]^*[D_1,\pi(a)]=(1+\Omega^2)\mathfrak{U}^\dag\begin{pmatrix}
L^*L&0&0&0\\
0&{\bar{L}}^*{\bar{L}}&0&0\\
0&0&L^*L&0\\
0&0&0&{\bar{L}}^*{\bar{L}}\end{pmatrix}\mathfrak{U}\label{diag1}
\end{equation}
with unitary $\mathfrak{U}$ given by
\begin{equation}
\mathfrak{U}=\begin{pmatrix}
1&0&0&0\\
0&(1+\Omega^2)^{-{{1}\over{2}}}&\Omega(1+\Omega^2)^{-{{1}\over{2}}}&0\\
0&-\Omega(1+\Omega^2)^{-{{1}\over{2}}}&(1+\Omega^2)^{-{{1}\over{2}}}&0\\
0&0&0&1\end{pmatrix}\label{diagprim1}.
\end{equation}
Then, \eqref{diag1}, \eqref{diagprim1} implies that for any $\mathbb{A}$, $\ell_{D_1}(a)^2=||[D_1,\pi(a)] ||^2=
||[D_1,\pi(a)]^*[D_1,\pi(a)] ||=(1+\Omega^2)\max(||L^*L ||,||{\bar{L}}^*{\bar{L}} ||)=(1+\Omega^2)\max(||L ||^2,||{\bar{L}} ||^2)=(1+\Omega^2)\ell_{D_0}^2(a)$. A similar result holds for $k=2$. Hence
\begin{equation}
\ell_{D_k}(a)=(1+\Omega^2)^{{{1}\over{2}}}\ell_{D_0}(a)\label{relatlipnorm}.
\end{equation}
Equation \eqref{relatlipnorm} implies that $\ell_{D_k}(a)=||[D_k,\pi(a)]||$ is bounded since $\ell_{D_0}$ is bounded by Proposition \ref{X0}. Hence, Axiom iii is satisfied so that $\mathfrak{X}(k)$, $k=1,2$ is a spectral triple.\\
Finally, one can write for any states $\omega_1,\omega_2\in\mathfrak{S}(\mathbb{A})$
\begin{align}
d_{D_k}(\omega_1,\omega_2)&=\sup\big\{\vert \omega_1(a)-\omega_2(a)\vert;\ a\in\mathbb{A},\ell_{D_k}(a)\le1\big\}\nonumber\\
&=\sup\big\{ {{\vert \omega_1(a)-\omega_2(a)\vert}\over{(1+\Omega^2)^{{{1}\over{2}}}\ell_{D_0}(a)}};\ a\in\mathbb{A}\big\},
\end{align}
where the last equality comes from \eqref{relatlipnorm}, from which follows \eqref{distanceD}. This terminates the proof.
\end{proof}

\begin{remark}
Theorem \ref{th1} leads to immediate consequences or observations.\\
First, from Theorems \ref{th0} and \ref{th1}, one infers
\begin{equation}
d_{D_k}(\omega_m,\omega_n)={\sqrt{{{\theta}\over{2(1+\Omega^2)}} }}\sum_{p=n+1}^{m}{{1}\over{{\sqrt{p}} }},\ n<m,\ \forall k=1,2\label{formula2},
\end{equation}
where $\omega_m$, for any $m\in\mathbb{N}$ are defined by \eqref{diagon-states}.\\
A connection to physics can be made by noticing that the spectral distance formulas obtained above express a distance between the energy eigenstates of the harmonic oscillator. \\
A careful comparison of the spectral distance on the standard Moyal plane $\mathfrak{X}_0$ to the notion of quantum distance introduced in the mathematical physics literature by the authors of \cite{DFR} has been performed in \cite{mart-tom} to which we refer for more details relevant to fundamental physics.\\
A spectral distance formula among coherent states (i.e the ``quantum points'' that we will use below in the study of causal curves) for the standard Moyal plane has been derived in \cite{mart-recent}. This, combined with Theorem \ref{th1} leads to the conclusion that the above spectral distances between two arbitrary coherent states are proportional to the Euclidean distance.
\end{remark}

\subsection{Metric properties of Moyal planes.}\label{subsection24}

$\mathfrak{X}_{D_0}$ and $(\mathfrak{X}_k)_{k=1,2}$ define representative examples of quantum locally compact metric space as defined by Latr\'emoli\`ere in \cite{latrem2}. The existence of pure states at infinite spectral distance to each other that will be considered in this subsection forbids the spectral distance for $\mathfrak{X}_{D_0}$ and $(\mathfrak{X}_k)_{k=1,2}$ to metrize the w*-topology on $\mathfrak{S}(\mathbb{A})$. Hence $\mathfrak{X}_{D_0}$ and $(\mathfrak{X}_k)_{k=1,2}$ cannot belong to the category of compact quantum metric space as defined by Rieffel.\par

We set from now on $d_{D_0}=d_0$, $d_{D_k}=d_k$.

\begin{definition}\label{psiss}
Let $\psi(s)$ be a family of unit vectors of $L^2(\mathbb{R}^2)$ defined by $\psi(s)=\sum_{m}\psi_m(s)f_{m0}={{1}\over{{{\sqrt{2\pi\theta}}}}} \sum_{m\in\mathbb{N}}{\sqrt{{{1}\over{\zeta(s)(m+1)^s}}}}f_{m0}$ for any $s\in\mathbb{R}$, $s>1$, where $\zeta(s)$ is the Riemann zeta function. The related pure states are denoted by $\omega_{\psi(s)}$, for any $s\in\mathbb{R}$, $s>1$, with $\omega_{\psi(s)}(a)=2\pi\theta\sum_{m,n}\psi^*_m(s)\psi_n(s)a_{mn}$.
\end{definition}
The following property holds.
\begin{proposition}\label{state-infinite}(\cite{Wal1},\cite{Wal2})
$d_k(\omega_n,\omega_{\psi(s)})=+\infty$, $\forall k=0,1,2$, $\forall \,s\in\,]1,{{3}\over{2}}]$, $\forall\, n\in\mathbb{N}$.
\end{proposition}
\begin{proof}
Define
\begin{equation}
B(m_0,\psi(s)):=|\omega_{0}(\hat{a}(m_0))-\omega_{\psi(s)}(\hat{a}(m_0))|
\end{equation}
Compute
\begin{eqnarray}
B(m_0,\psi(s))&=&{\sqrt{{{\theta}\over{2}}}}\left|\sum_{m=0}^{m_0}\sum_{k=m}^{m_0}{{1}\over{{\sqrt{k+1}} }}{{1}\over{\zeta(s)(m+1)^s }}-\sum_{k=0}^{m_0}  {{1}\over{{\sqrt{k+1}} }} \right|\nonumber\\
&=&{\sqrt{{{\theta}\over{2}}}}\,\Bigg|\left(1-{{1}\over{\zeta(s)}}
\sum_{m=0}^{m_0}{{1}\over{(m+1)^s}}\right)\left(\sum_{k=0}^{m_0}{{1}\over{{\sqrt{k+1}} }}\right)\nonumber\\
&+&{{1}\over{\zeta(s)}}\sum_{m=0}^{m_0}\sum_{k=0}^{m-1}{{1}\over{(m+1)^s{\sqrt{k+1}} }} \Bigg|.
\label{finalbound}
\end{eqnarray}
$B(m_0,\psi(s))$ is thus the sum of 2 positive terms: $B(m_0,\psi(s))={\sqrt{{{\theta}\over{2}}}}|A_1(m_0)+A_2(m_0) |$. Now, observe that
\begin{equation}
A_2(m_0)={{1}\over{\zeta(s)}}\sum_{m=0}^{m_0}\sum_{k=0}^{m-1}{{1}\over{(m+1)^s{\sqrt{k+1}} }}\ge
{{2}\over{\zeta(s)}}\sum_{m=0}^{m_0}({{{\sqrt{m+1}}-1) }\over{(m+1)^s}},
\end{equation}
where we used the estimate 
\begin{equation}
\sum_{k=0}^{m}{{1}\over{{\sqrt{ k+1}} }}\ge 2({\sqrt{m+2}}-1). \nonumber
\end{equation}
But when $s\le{{3}\over{2}}$, one has $\lim\limits_{m_0\to+\infty}A_2(m_0)=+\infty$. Thus, $\lim\limits_{m_0\to+\infty}B(m_0,\psi(s))=+\infty$. One concludes $d_k(\omega_0,\omega_{\psi(s)})=+\infty$, $\forall\, s\in\,]1,{{3}\over{2}}]$. By the triangular inequality $d_k(\omega_0,\omega_{\psi(s)})\le d_k(\omega_0,\omega_n)+d_k(\omega_n,\omega_{\psi(s)})$, for any $n\in\mathbb{N}$. But $d_k(\omega_0,\omega_n)$ is finite. Hence, $d_k(\omega_n,\omega_{\psi(s)})=+\infty$, $\forall\, s\in\,]1,{{3}\over{2}}]$, for any $n\in\mathbb{N}$.
\end{proof}
The distance between any pair of states $\omega_{\psi(s)}$'s is infinite. Namely:
\begin{proposition}\label{infinite-points}(\cite{Wal1})
$d_k(\omega_{\psi(s_1)},\omega_{\psi(s_2)})=+\infty$, for any  $k=0,1,2$, $s_1,s_2\in \,]1,{{5}\over{4}}[   \cup  ]{{5}\over{4}},{{3}\over{2}}]$, $s_1\ne s_2$.
\end{proposition}
\begin{proof}
The lengthy proof is given in \cite{Wal1}. It uses the mean value theorem to obtain suitable successive estimates.
\end{proof}

\begin{remark}\label{finitestate}
let $\mathcal{I}$ be a finite subset of $\mathbb{N}$ and let $\Lambda=\sum_{m\in \mathcal{I}}\lambda_me_{m}$ denotes a unit vector of $\ell^2(\mathbb{N})$ with corresponding state $\omega_\Lambda$. The spectral distances between any state $\omega_n$ and $\omega_\Lambda$ is finite,
$d_k(\omega_n,\omega_\Lambda)<\infty$, for any $k=0,1,2$ and any $n\in\mathbb{N}$. Indeed, compute
\begin{eqnarray}
|\omega_\Lambda(a)-\omega_n(a)|
&=&\Bigg|2\pi\theta\Bigg(\sum_{p,q\in\mathcal{I}}a_{pq}\lambda_p^*\lambda_q\Bigg)-a_{nn}\Bigg|\nonumber\\
&\le& 2\pi\theta\Bigg(\sum_{p,q\in\mathcal{I}}|a_{pq}||\lambda_p^*\lambda_q|\Bigg)+|a_{nn}|\le
\sum_{p,q\in\mathcal{I}}|a_{pq}|+2\pi\theta|a_{nn}| \label{distancefinite-quant}.
\end{eqnarray}
Assume now $k=0$. By using eqn.(8) of Proposition 5 in \cite{Wal1}, one infers that the $a_{mn}$ are finite for any element in the Lipschitz ball $\mathfrak{B}_{\ell_{D_0}}$. Hence $d_0(\omega_n,\omega_\Lambda)<\infty$ which extends to $k=1,2$ by Theorem \ref{th1}.
\end{remark}

\begin{proposition}\label{connected1}
For any $k=0,1,2$, the spectral triple $\mathfrak{X}(k)$ defines a quantum space with an infinite number of distinct connected components, each component being pathwise connected for the $d_k$-topology.
\end{proposition}
\begin{proof}
Assume $s_1\in ]1,{{5}\over{4}}[   \cup  ]{{5}\over{4}},{{3}\over{2}}]$. Define 
\begin{equation}
\mathfrak{S}^k_{\psi_{s_1}}:=\{\omega\in\mathfrak{S}(\mathbb{A}),\ d_k(\omega,\ \omega_{\psi_{s_1}})<+\infty\}. 
\end{equation}
Let $B(\omega,\rho)\subset\mathfrak{S}(\mathbb{A})$, the open ball with center $\omega$ and radius $\rho>0$. For any element $\eta$ in $B(\omega,\rho)\subset\mathfrak{S}(\mathbb{A})$ and for any $\omega\in\mathfrak{S}^k_{\psi_{s_1}}$, one has 
\begin{equation}
d_k(\eta,\omega_{\psi_{s_1}})\le d_k(\eta,\omega)+d_k(\omega,\omega_{\psi_{s_1}})<+\infty, 
\end{equation}
which holds true for $k=0,1,2$. Thus, $\mathfrak{S}^k_{\psi_{s_1}}$ is open for any $k=0,1,2$.\\
For any $\omega$ in the complement of $\mathfrak{S}^k_{\psi_{s_1}}$ and any $\eta\in\mathfrak{B}_\rho(\omega)$, one has 
\begin{equation}
d_k(\omega,\omega_{\psi_{s_1}})\le d_k(\eta,\omega_{\psi_{s_1}})+d_k(\eta,\omega).
\end{equation}
Thus, $d_k(\eta,\omega_{\psi_{s_1}})=+\infty$ and therefore $\eta\in\mathfrak{S}_{\psi_{s_1}}$. Hence $\mathfrak{S}^k_{\psi_{s_1}}$ is also closed. Hence, $\mathfrak{S}^k_{\psi_{s_1}}$, $k=0,1,2$ is a closed-open subset of $\mathfrak{S}(\mathbb{A})$ which is therefore an union of connected parts while $\omega_{\psi_s}\notin\mathfrak{S}^k_{\psi_{s_1}}$, $\forall s\ne s_1$ and cannot belong to the same connected component.\\
Pathwise connectedness follows from the fact that the map $\omega_t:\ t\in[0,1]\to\mathfrak{S}(\mathbb{A} )$ defined by $\omega_t:=(1-t)\omega_1+t\omega_2$, for any $\omega_1,\omega_2\in\mathfrak{S}_{\psi_{s_1}}$ is $d_k$-continuous, since one has $d_k(\omega_{t_1},\omega_{t_2})=|t_1-t_2|d_k(\omega_1,\omega_2)$ which is readily obtained from the very definition of $d_k$ and $d_k(\omega_1,\omega_2)<+\infty$.
\end{proof}
\begin{remark}
From the above discussion, it follows that any of the $\mathfrak{X}_k$, $k=0,1,2$ defines a quantum space with infinite diameter. 
\end{remark}
The quantum spaces $(\mathfrak{X}_k)_{k=0,1,2}$ define quantum locally compact metric spaces as introduced by Latr\'emoli\`ere in \cite{latrem2}. In order to make contact with \cite{latrem2}, we identify the relevant structures needed in the general construction. First, $(\bar{\mathbb{A}}, l_k)$ defines obviously a {\it{Lipschitz pair}} as stated in Definition 2.3 \cite{latrem2} where the seminorm $l_k$ has domain $\text{Dom}(l_k)=\mathbb{A}$ dense in $\mathbb{A}$ and $l_k(a)=0$ when $a=\lambda\bbone$, $\lambda\in\mathbb{C}$. Let $\mathfrak{D}\subset\bar{\mathbb{A}}$ be the C*-subalgebra generated by the diagonal vectors of the matrix basis, i.e $(f_{mm})_{m\in\mathbb{N}}$. By Proposition \ref{approx-unit}, $\mathfrak{D}$ involves the approximate unit. Hence $(\bar{\mathbb{A}},\mathfrak{D})$ is a {\it{topographic quantum space}} as in Definition 2.15 \cite{latrem2} so that the data $(\bar{\mathbb{A}},l_k,\mathfrak{D})$ define a Lipschitz triple, Definition 2.27 \cite{latrem2}. In the terminology of \cite{latrem2}, the spectral distance defined e.g in \eqref{absdistance} is called the {\it{extended Monge-Kantorovitch metric}} for the Lipschitz pair $(\bar{\mathbb{A}},l_k)$ denoted by $mk_{l_k}(\omega,\eta)$, Definition 2.4 in \cite{latrem2} for any $\omega,\eta\in\mathfrak{S}(\bar{\mathbb{A}})$.\par
There are two additional notions to be used \cite{latrem2}. First, the notion of tame sets for a Lipschitz triple, subsets of $\mathfrak{S}(\bar{\mathbb{A}})$, can be viewed as a topological condition providing a convenient way to always obtain a natural noncommutative analog of the notion of tight set in probability theory. The second one is the notion of local state space. For technical reasons, it is convenient to require further regularity condition for the Lipschitz triple, namely that any subset $\mathfrak{S}(\bar{\mathbb{A}}|K)$ involving the restricted states to the compact set $K\in\mathcal{K}(\sigma(\mathfrak{D}))$ 
($\sigma(\mathfrak{D})$ is the spectrum of $\mathfrak{D}$) has finite diameter for the extended Monge-Kantorovitch metric of the Lipschitz pair. This additional condition ensures that the definition of tame set does not depend on any choice of a local state and that tame sets are involved in closed balls of finite radius around any local state. Then, one defines
\begin{definition}(Definition 3.1 \cite{latrem2})
A quantum locally compact metric space is a regular Lipschitz triple for which the topology of the metric space $(\mathfrak{K}, mk_L)$ for any tame set $\mathfrak{K}$ is the relative topology induced by the w*-topology restricted on $\mathfrak{K}$, where $mk_L$ is the corresponding Monge-Kantorovitch metric.
\end{definition}
We quote the result of \cite{latrem2}:
\begin{theorem}(Theorem 4.9 \cite{latrem2})
$\mathfrak{X}_0$ defines a quantum locally compact (separable) metric space.
\end{theorem}
This result can be immediately completed by using Theorem \ref{th1}.
\begin{proposition}
$\mathfrak{X}_k$ for any ${k=1,2}$ defines a quantum locally compact (separable) metric space.
\end{proposition} 
\begin{proof}
Pick the local state $\mu(a)=\langle e_0,ae_0\rangle=a_{00}$ for any $a\in\bar{\mathbb{A}}$. The indicator function on any compact set $K\subset\sigma(\mathfrak{D})$ is $\chi_K=\sum_{m\in K}f_{mm}$, $K$ a finite subset of $\mathbb{Z}$. Define $\mathfrak{L}_1(\bar{\mathbb{A}},\ell_k,\mathfrak{D})=\{a\in\bar{\mathbb{A}}\ ;\ l_k(a)\le1,\mu(a)=0 \}$,  for any $k=0,1,2$. By Theorem 4.6 \cite{latrem2},  $\chi_K\mathfrak{L}_1(\bar{\mathbb{A}},\ell_0,\mathfrak{D})\chi_K$ is precompact. But by Theorem \ref{th1}, $l_k=(1+\Omega^2)^{\frac{1}{2}}l_0$, $k=1,2$. Thus $\chi_K\mathfrak{L}_1(\bar{\mathbb{A}},\ell_k,\mathfrak{D})\chi_K$ is precompact. The proposition then follows from Theorem 3.9 \cite{latrem2}.
\end{proof}
Other interesting examples of quantum locally compact quantum metric spaces are provided by a family of noncommutative spaces related to the space $\mathbb{R}^3_\lambda$ pertaining to the mathematical physics litterature \cite{vit-wal}. This will be presented in a future publication.

\section{Presence of causality on Moyal plane with Minkowski metric}\label{sec3}

\subsection{Lorentzian spectral triples and causality in noncommutative geometry}\label{subsection31}

The notion of spectral triple as presented in the previous section is mainly used in an Euclidean context, i.e.~on manifolds with Riemannian signature, and most of the applications of Connes' noncommutative geometry to Moyal planes have been done using this signature exclusively. As an emerging branch of the theory, Lorentzian noncommutative geometry is an attempt to adapt the main components of noncommutative geometry to manifolds with Lorentzian signature. While its development is far for being complete, there is enough material to make applications in the domain of mathematical physics. In particular, Lorentzian noncommutative geometry allows us to define a notion of causality in noncommutative geometry  \cite{CQG2013}. This notion has already been applied to some specific models of almost-commutative manifolds \cite{SIGMA2014,Franco2015}. In this section, we will make the first exploration of this notion on the 2-dimensional Moyal plane by switching from the Euclidean metric to the Minkowski metric. Such a study is quite important since the presence of causality on Moyal plane is controversial in the domain of quantum field theory \cite{Balachandran, Bahns}. The main problem comes from the fact that noncommutative spaces are non-local and the usual notion of point cannot be used. We will show that causal relations are possible on Moyal plane if we consider specific pure states on the algebra, which are in fact Gaussian functions. The causal structure within those states is completely similar to the causal structure on the usual Minkowski space, except that the notion of locality is lost with the noncommutative algebra.

A spectral triple with Lorentzian signature is not so much different from the Riemannian notion, except that the Dirac operator is naturally self-adjoint in a Krein space (a space with indefinite inner product and some specific conditions \cite{Bog,Stro}) instead of a Hilbert space. A Hilbert space can still be used by considering a specific operator $\mathcal{J}$ called fundamental symmetry, which turns the Krein space into a Hilbert space and vice versa. In the specific case of Moyal plane, this operator is just the first gamma matrix $\gamma^0$. There exist different but compatible definitions of Lorentzian spectral triples (see e.g.~\cite{Rennie12, Rennie15}) and we will use a specific one which is a particular case of the others. The advantage of this definition is that no signature other than the Lorentzian one is allowed \cite{Franco2014} so a notion of causality is always well defined.

\begin{definition}
\label{deflost}
A Lorentzian spectral triple is given by $(\mathbb{A},\widetilde{\mathbb{A}},\pi,\mathcal{H},D,\mathcal{J})$
with:
\begin{itemize}
\item A~Hilbert space $\mathcal{H}$.
\item A~non unital pre-C*-algebra $\mathbb{A}$ with a~faithful *-representation $\pi$ on
$\mathcal{B}(\mathcal{H})$.
\item A~preferred unitization $\widetilde{\mathbb{A}}$ of $\mathbb{A}$, which is also a pre-C*-algebra,
with a compatible faithful *-representation $\pi$  on
$\mathcal{B}(\mathcal{H})$ and such that $\mathbb{A}$ is an ideal of
$\widetilde{\mathbb{A}}$.
\item An unbounded operator $D$, densely defined on $\mathcal{H}$, such that:
\begin{itemize}
\item $\forall a\in\widetilde{\mathbb{A}}$,   $[D,\pi(a)]$ extends to a~bounded operator on $\mathcal{H}$, \item $\forall
a\in\mathbb{A}$,   $\pi(a)(1 + \left< D\right >^2)^{-\frac 12}$ is compact, with $\left< D \right >^2 := \frac 12 (D D^* + D^*
D)$.
\end{itemize} 
\item A~bounded operator $\mathcal{J}$ on $\mathcal{H}$ with $\mathcal{J}^2=1$, $\mathcal{J}^*=\mathcal{J}$, $[\mathcal{J},\pi(a)]=0$,
$\forall a\in\widetilde{\mathbb{A}}$ and such that:
\begin{itemize}
\item $D^*=-\mathcal{J} D \mathcal{J}$ on $\text{Dom}(D) = \text{Dom}(D^*) \subset \mathcal{H}$;
\item there exists a densely defined self-adjoint operator $\mathcal{T}$ with $\text{Dom}(\mathcal{T})$ $\cap$ $\text{Dom}(D)$ dense in $\mathcal{H}$ and with $\left(1+ \mathcal{T}^2 \right)^{-\frac{1}{2}}\in
\widetilde{\mathbb{A}}$, and a positive element $N\in\widetilde{\mathbb{A}}$ such that $\mathcal{J}  = -N [D,\mathcal{T}]$.
\end{itemize}
\end{itemize}
\end{definition}

\begin{definition}
We say that a Lorentzian spectral triple is even if there exists a~$\mathbb Z_2$-grading $\gamma$ of $\mathcal{H}$ such that
$\gamma^*=\gamma$, $\gamma^2=1$, $[\gamma,\pi(a)] = 0$ $\forall a\in\widetilde{\mathbb{A}}$, $\gamma
\mathcal{J} =- \mathcal{J} \gamma$ and $\gamma D =- D \gamma $.
\end{definition}

We must notice that we will always work with an algebra $\mathbb{A}$ which is non unital, since unital algebras correspond to compact manifolds, and the notion of causality cannot be well defined on compact Lorentzian manifolds (without boundaries). The role of the unitization $\widetilde{\mathbb{A}}$ is in fact purely technical, but is a need for the following definition:

\begin{definition}[\cite{CQG2013}]\label{defncauscone}
Let $\mathcal{C}$ be the convex cone of all Hermitian elements $a \in \widetilde{\mathbb{A}}$ respecting
\begin{equation}\label{constcaus}
\forall \; \phi\in\mathcal{H},\qquad\left<\phi,\mathcal{J}[D,\pi(a)]\phi\right>\leq0,
\end{equation}
where $\left<\cdot,\cdot\right>$ is the inner product on $\mathcal{H}$.
If the following condition is fulfilled:
\begin{equation}\label{condcausal}
\overline{span_{{\mathbb C}}(\mathcal{C})}=\overline{\widetilde{\mathbb{A}}},
\end{equation}
then $\mathcal{C}$ is called a causal cone. It induces a partial order relation on
$\mathfrak{S}(\mathbb{A})$, which we call causal relation, by:
\begin{equation}
\forall \omega,\eta\in \mathfrak{S}(\mathbb{A}),\qquad\omega\preceq\eta\qquad\text{iff}\qquad\forall a\in\mathcal{C},\quad\omega(a)\leq\eta(a).
\end{equation}
\end{definition}

This definition brings a notion of causality valid for every Lorentzian spectral triple, even when the algebra is noncommutative. The causality must be understood as a partial order relation between the states of the algebra, which can be restricted to pure states only (but this is not mandatory). In the commutative regime, the causal cone $\mathcal{C}$ is exactly the set of smooth causal functions, which are the smooth functions non-decreasing along every future directed causal curve. When the manifold, on which a commutative Lorentzian spectral triple is based, is globally hyperbolic, the complete set $\mathcal{C}$ is sufficient to characterize the causal structure \cite{CQG2013,Bes09} (the condition of global hyperbolicity is sufficient but not necessary). In the noncommutative regime, the sufficient condition to have a well defined partial order among all states of the algebra is the condition \eqref{condcausal}. Here the technical role of the unitization $\widetilde{\mathbb{A}}$ is clear, since there is no non-trivial monotonic function in $\mathbb{A}=C^\infty_0(M)$ in the commutative case. However one must be careful that the states in $\mathfrak{S}(\mathbb{A})$ should extend in a unique way in $\mathfrak{S}(\widetilde{\mathbb{A}})$. This is the case for every commutative manifold, almost-commutative manifold and Moyal plane.

The name of causal relation in Definition \ref{defncauscone} is completely justified by the fact that this relation corresponds to the usual one when the algebra $\mathbb{A}$ is commutative:

\begin{theorem}[\cite{CQG2013}]
Let $(\mathbb{A},\widetilde{\mathbb{A}},\pi,\mathcal H,D,\mathcal J)$ be a Lorentzian spectral triple with a commutative algebra $\mathbb{A}=C^\infty_0( M)$ constructed from a complete globally hyperbolic manifold $ M$, then the causal structure defined by Definition \ref{defncauscone}, if restricted to the pure states on $\mathbb{A}$, corresponds exactly to the usual causal structure on $ M$, using the Gelfand correspondence between the pure states and the points of $ M$.
\end{theorem}

A non technical review of the notion of causality in noncommutative geometry and some applications can be found in \cite{CC2014}.

\subsection{The Moyal Lorentzian spectral triple}\label{subsection32}

On the Minkowski space ${\mathbb R}^{1,1}$, the Moyal Lorentzian spectral triple $(\mathbb{A},\widetilde{\mathbb{A}},\pi,\mathcal{H}_0,D,\mathcal{J})$ is constructed in the following way:
\begin{itemize}
\item $\mathcal{H}_0 := L^2({\mathbb R}^{1,1}) \otimes {\mathbb C}^{2}$ is the Hilbert space of square integrable sections of the spinor bundle over the two-dimensional Minkowski space-time with the usual positive definite inner product $\langle \psi,\phi\rangle=\int d^2x\ (\psi_1^*\phi_1+\psi_2^*\phi_2) $ $\forall\, \psi,\phi\in{\mathcal{H}_0}$ with $\psi=(\psi_1,\psi_2)$, $\phi=(\phi_1,\phi_2)$.

\item $\mathbb{A}$ is the space of Schwartz functions $\mathcal{S} = \mathcal{S}({\mathbb R}^{1,1})$ with the Moyal $\star$ product. The representation $\pi:{\mathbb{A}}\to{\mathcal{B}}({\mathcal{H}_0})$ is defined by the left multiplication:
\begin{equation}\pi(a)=L(a)\otimes\mathbb{I}_2,\qquad \pi(a)\psi=(a\star\psi_1,a\star\psi_2).\end{equation}

\item $\widetilde{\mathbb{A}}$ is some preferred unitization of $\mathbb{A}$ which must be a sub-algebra of the multiplier algebra ${\mathcal{M}}(\mathcal A)=\{a\in{\mathcal{S}}^\prime\ /\ a\star b\in{\mathcal{S}},\ b\star a\in{\mathcal{S}},\ \forall b\in{\mathcal{S}} \}$. A typical choice is $\widetilde{\mathbb{A}} = (\mathcal{B},\star) \subset \mathcal{M}(\mathcal A)$ the unital Fr\'echet pre-C*-algebra of smooth functions which are bounded together with all derivatives \cite{GAYRAL2004}. However, we will consider a bigger (unbounded) algebra in the following for a technical reason.
\item $D := -i \partial_\mu \otimes \gamma^\mu$ (with $\mu = 0,1$) is the flat Dirac operator on ${\mathbb R}^{1,1}$ where:
\begin{equation}\gamma^0 = i\sigma^1=\begin{pmatrix} 0&i \\ i& 0 \end{pmatrix},\qquad \gamma^1 = \sigma^2=\begin{pmatrix} 0&i \\ -i& 0 \end{pmatrix}\end{equation}
are the flat Dirac matrices which verify $\gamma^\mu\gamma^\nu+\gamma^\nu\gamma^\mu=2\eta^{\mu\nu}$, $\forall\, \mu,\nu=0,1$ (we use $(-,+)$ as convention for the signature of the metric).
\item $\mathcal{J}:= i\gamma^0$ is the fundamental symmetry which turns the Hilbert space $\mathcal{H}_0$ into a Krein space.
\end{itemize}

As proved in \cite{Franco2014}, this construction respects all the axioms of a Lorentzian spectral triple. Such a kind of construction has already been used in the context of quantum field theory \cite{Verch11}.\\

If we define $\partial_+:=\partial_0+\partial_1$ and  $\partial_-:=\partial_0-\partial_1$, we can write the Dirac operator as:
\begin{equation}D=\begin{pmatrix} 0&\partial_+ \\ \partial_-& 0 \end{pmatrix}\cdot\end{equation}

Then the operator $\mathcal{J}[D,\pi(a)]$ of the causal constraint \eqref{constcaus} is simply:
\begin{eqnarray}
 \mathcal{J}[D,\pi(a)] &=&  \mathcal{J} D \pi(a) - \mathcal{J} \pi(a) D \nonumber\\
 &=& \begin{pmatrix} -\partial_- L(a) + L(a) \partial_- & 0 \\ 0 & -\partial_+ L(a)  +L(a) \partial_+ \end{pmatrix} \nonumber\\
 &=& - \begin{pmatrix} L(\partial_- a) & 0 \\ 0 & L(\partial_+ a) \label{expressD}\end{pmatrix}
 \end{eqnarray}
 
 \begin{proposition} Using the matrix basis, for every $a = \sum_{mn} a_{mn} f_{mn}$ $\in \widetilde{\mathbb{A}}$, if we define $\partial_- a =  \sum_{mn} \alpha_{mn} f_{mn} $ and $\partial_+ a =  \sum_{mn} \beta_{mn} f_{mn} $, then:
 \begin{enumerate}\itemsep=10pt
\item[a)] The following relations hold, with $\lambda := \frac{1+i}{\sqrt2}$:
\begin{equation}\label{expressalpha}
\alpha_{mn} = a_{m+1,n} \lambda \sqrt{\frac{m+1}{\theta}} + a_{m,n+1} \bar\lambda \sqrt{\frac{n+1}{\theta}} - a_{m-1,n} \bar\lambda \sqrt{\frac{m}{\theta}} - a_{m,n-1} \lambda \sqrt{\frac{n}{\theta}},
\end{equation}
\begin{equation}\label{expressbeta}
\beta_{mn} = a_{m+1,n} \bar\lambda \sqrt{\frac{m+1}{\theta}} + a_{m,n+1} \lambda \sqrt{\frac{n+1}{\theta}} - a_{m-1,n} \lambda \sqrt{\frac{m}{\theta}} - a_{m,n-1} \bar\lambda \sqrt{\frac{n}{\theta}}.
\end{equation}
\item[b)] The following conditions are equivalent:
\begin{equation}
\forall \; \phi\in\mathcal{H}_0,\quad\left<\phi,\mathcal{J}[D,\pi(a)]\phi\right>\leq0\qquad(\Leftrightarrow a \in\mathcal{C})
\end{equation}
\begin{equation}\Longleftrightarrow\nonumber\end{equation}
\begin{equation}
(\alpha_{mn})_{mn\in{\mathbb N}} \text{ and } (\beta_{mn})_{mn\in{\mathbb N}} \text{ are  semi-positive definite (infinite) matrices.} \nonumber
\end{equation} 
\end{enumerate}
 
 \end{proposition}

\begin{proof}
Since we are using the indices $\mu=0,1$ for the coordinates in Lorentzian signature, the matrix basis looks like:
\begin{equation}f_{mn}={{1}\over{(\theta^{m+n}m!n!)^{1/2}}}{\bar{z}}^{\star m}\star f_{00}\star z^{\star n} \end{equation}
with
\begin{equation}{\bar{z}}={{1}\over{{\sqrt{2}}}}(x_0-ix_1), z={{1}\over{{\sqrt{2}}}}(x_0+ix_1) \text{ and } f_{00} = 2e^{-\frac{x_0^2+x_1^2}{\theta}}.\end{equation}
From \cite{Gracia-Bondia:1987kw} we know that $\bar z \star f_{00} = 2 \bar z f_{00}$ and  $ f_{00} \star z = 2 z f_{00} $. Since $x_0 + x_1 = \lambda \bar z + \bar \lambda z$, we have 
\begin{equation}\partial_0 f_{00} + \partial_1 f_{00} = -\frac{2(x_0+x_1) f_{00}}{\theta} = -\frac{\lambda \bar z \star f_{00} + \bar \lambda f_{00} \star z}{\theta}  .\end{equation} Then we can compute the derivatives of $f_{mn}$:
\begin{eqnarray}\label{derivfmn}
\partial_+ f_{mn} &=& \partial_0 f_{mn} + \partial_1 f_{mn} \nonumber\\
&=& \frac{m}{\sqrt2}\;  {{1}\over{(\theta^{m+n}m!n!)^{1/2}}} {\bar{z}}^{\star {m-1}}\star f_{00}\star z^{\star n} \nonumber\\
&-&  i\frac{m}{\sqrt2}\; {{1}\over{(\theta^{m+n}m!n!)^{1/2}}} {\bar{z}}^{\star {m-1}}\star f_{00}\star z^{\star n} \nonumber\\
&+& \frac{n}{\sqrt2}\;  {{1}\over{(\theta^{m+n}m!n!)^{1/2}}} {\bar{z}}^{\star {m}}\star f_{00}\star z^{\star {n-1}} \nonumber\\
&+& i\frac{n}{\sqrt2}\;  {{1}\over{(\theta^{m+n}m!n!)^{1/2}}} {\bar{z}}^{\star {m}}\star f_{00}\star z^{\star {n-1}} \nonumber\\
&-& \frac{\lambda}{\theta}\;  {{1}\over{(\theta^{m+n}m!n!)^{1/2}}} {\bar{z}}^{\star {m+1}}\star f_{00}\star z^{\star n} \nonumber\\
&-& \frac{\bar\lambda}{\theta}\;  {{1}\over{(\theta^{m+n}m!n!)^{1/2}}} {\bar{z}}^{\star {m}}\star f_{00}\star z^{\star {n+1}} \nonumber\\
&=& \bar \lambda \sqrt{\frac{m}{\theta}} f_{m-1,n} +  \lambda \sqrt{\frac{n}{\theta}} f_{m,n-1} -  \lambda \sqrt{\frac{m+1}{\theta}} f_{m+1,n} - \bar \lambda \sqrt{\frac{n+1}{\theta}} f_{m,n+1}; \nonumber\\
&&\\
\partial_- f_{mn} &=&  \lambda \sqrt{\frac{m}{\theta}} f_{m-1,n} + \bar \lambda \sqrt{\frac{n}{\theta}} f_{m,n-1} - \bar \lambda \sqrt{\frac{m+1}{\theta}} f_{m+1,n} -  \lambda \sqrt{\frac{n+1}{\theta}} f_{m,n+1}. \nonumber\\
&&
\end{eqnarray}
The relations (a) follow from the identification of the coefficients of the $f_{mn}$ in the development of $\partial_+ a =  \sum_{mn} \beta_{mn} f_{mn} = \sum_{mn} a_{mn} \partial_+ f_{mn}$ and $\partial_- a =  \sum_{mn} \alpha_{mn} f_{mn} = \sum_{mn} a_{mn} \partial_- f_{mn}$.\\

Using the formulation \eqref{expressD}, the causality condition  $\forall \, \phi=(\phi_1,\phi_2)\in\mathcal{H}_0$,  \mbox{$\left<\phi,\mathcal{J}[D,\pi(a)]\phi\right>\leq0$} is equivalent to:
\begin{equation} \forall \phi_1\in L^2({\mathbb R}^{1,1}) , \int  d^2x\ \phi_1^*( (\partial_- a) \star \phi_1)  = \int d^2x\ \phi_1^* \star (\partial_- a) \star \phi_1  \geq 0 \end{equation}
and
\begin{equation} \forall \phi_2\in L^2({\mathbb R}^{1,1}) , \int  d^2x\  \phi_2^*( (\partial_+ a) \star \phi_2)  =  \int d^2x\  \phi_2^*\star (\partial_+ a) \star \phi_2 \geq 0.\end{equation}

We need to check that the semi-positive definiteness of an operator using the Moyal left multiplication is equivalent to the usual semi-positive definiteness of a matrix. 

Using the matrix basis $\phi_1 = \sum_{mn} \phi_{mn} f_{mn} $  and $ \partial_- a =  \sum_{mn} \alpha_{mn} f_{mn} $, we get:
\begin{eqnarray}
\int\! d^2x\; \phi_1^* \star (\partial_- a) \star \phi_1 &=&  \int\! d^2x \left(\sum_{mn} \bar \phi_{mn} \bar f_{mn} \right) \!\star\! \left(\sum_{kl} \alpha_{kl} f_{kl} \right) \!\star\!  \left(\sum_{qr} \phi_{qr} f_{qr} \right ) \nonumber\\
&=& \sum_{mnklqr}  \bar \phi_{mn} \alpha_{kl}\phi_{qr}  \int d^2x\ \bar f_{mn}  \star  f_{kl} \star f_{qr}  \nonumber\\
&=& \sum_{mnklqr}  \bar \phi_{mn} \alpha_{kl}\phi_{qr} \;(2\pi\theta)\; \delta_{mk}  \delta_{nr}  \delta_{lq} \nonumber\\ 
&=& (2\pi\theta) \sum_{mnl}  \bar \phi_{mn} \alpha_{ml}\phi_{ln} \geq 0. 
\end{eqnarray}

The last term contains a sum over $n$ of inner products $\left<\phi_{\cdot n}, (\alpha_{ml})_{ml\in{\mathbb N}} \,\phi_{\cdot n}\right>$, and since it must be valid for every $\phi_1= \sum_{mn} \phi_{mn} f_{mn} $, it must be valid for every infinite vector $\phi_{\cdot n}\in L^2({\mathbb N})$ with $n$ fixed. The same reasoning can be done for $\partial_+ a$.  Hence the condition is equivalent to $(\alpha_{mn})_{mn\in{\mathbb N}}$ and $(\beta_{mn})_{mn\in{\mathbb N}} $ be  semi-positive definite  matrices.
\end{proof}

\subsection{The causal structure between coherent states}\label{subsection33}

The pure states of ${{ {\mathbb{A}}}}$, as characterized in Proposition \ref{purestatesmoyal}, are the vector states in the matrix basis. Since this space is really huge, in order to find some causal relations within it we will restrict ourselves to a specific kind of pure state:

\begin{definition}
The coherent states of $\mathbb{A}$ are the vector states defined by:
\begin{equation}\varphi_m := \frac{1}{\sqrt{2\pi\theta}}e^{-\frac{{\left\vert \kappa\right\vert }^2}{2\theta}} \frac{\kappa^m}{\sqrt{m!\theta^m}},\end{equation}
 for any $\kappa\in{\mathbb C}$.
\end{definition}
The coherent states correspond to the possible translations under the complex scalar $\sqrt{2}\kappa$ of the ground state of the harmonic oscillator (i.e.~the vector state $\varphi_m = \frac{1}{\sqrt{2\pi\theta}} \delta_{m0}$ $\Leftrightarrow$ $\omega_\varphi = f_{00}$ which is a Gaussian function), using the correspondence $\kappa\in{\mathbb C} \cong {\mathbb R}^{1,1}$ with $a+ib \sim (a,b)$ \cite{mart-recent}. They are the states that minimize the uncertainty equally distributed in position and momentum. The classical limit of the coherent states, when $\theta \rightarrow 0$, corresponds to the usual pure states on ${\mathbb R}^{1,1}$, hence to the points of the usual Minkowski space.

Definition \ref{defncauscone} requires the setting of a specific unitization of the C*-algebra $\mathbb{A}$, and the considered states for the causal relation are those defined on this unitization. Since pure states on $\mathbb{A}$ are vector states, they are still well (and uniquely) defined on any unitization $\widetilde{\mathbb{A}}$ as long as their evaluation on the whole algebra is finite. The usual unitization $\widetilde{\mathbb{A}} = (\mathcal{B},\star)$ is not convenient for our purpose, since we will need the use of some linear functions in order to make the computation easier. However, since we are only interested by coherent states, we are free to chose a larger (even unbounded) unitization $\widetilde{\mathbb{A}}$ as long as all the coherent states are still well defined on it and $\widetilde{\mathbb{A}} \subset \mathcal{M}(\mathcal A)$ is still a *-subalgebra of the multiplier algebra. One must be careful that the representation of some elements of such an algebra is not necessarily a bounded operator, but this will just correspond to an unbounded infinite matrix, so every characterization using the matrix basis will still make sense. Also, the algebra must technically be chosen such that the condition \eqref{condcausal} is respected in order to guarantee that the partial order relation is well defined on the whole space of states (i.e.~the causal cone is sufficient to separate every state of the algebra). Once more, since we are only interested by coherent states we do not need to check this condition as long as the causal structure obtained between the coherent states is a well defined partial order relation.

Hence, in the following, we will define $\widetilde{\mathbb{A}}$ to be the largest *-subalgebra of the multiplier algebra $\mathcal{M}(\mathcal A)$ such that all coherent states are well defined on it. This algebra contains all smooth functions bounded together with all derivatives but also polynomials as shown by the following proposition:

\begin{proposition}\label{coherentstatesarefinite}
Let $\omega_\varphi$ be a coherent state corresponding to the complex scalar $\kappa$. Then for every $q\in\mathbb{N}$, $\omega_\varphi(z^q)$ is finite.
\end{proposition}
\begin{proof}
We have $z^q = \sum_{mn} a_{mn} f_{mn} $ with
\begin{eqnarray}
a_{mn} &=& \frac{1}{2\pi\theta} \left<z^q,f_{mn}\right>\nonumber \\
&=& \frac{1}{2\pi\theta}  \sqrt{\theta^q} \sqrt{\frac{(m+q)!}{m!}} \int d^2x\ f_{m+q,n} \nonumber  \\
&=&  \sqrt{\theta^q} \sqrt{\frac{(m+q)!}{m!}} \delta_{m+q,n}.
\end{eqnarray}
Hence the evaluation of $\omega_\varphi$ on $z^q$ is worth:
\begin{eqnarray}
\omega_\varphi(z^q) &=& 2\pi\theta \sum_{mn} \bar \varphi_m \varphi_n  a_{mn} \nonumber\\
&=&  \sum_{mn}  e^{-\frac{{\left\vert  \kappa\right\vert }^2}{\theta}}  \frac{\bar \kappa^m \kappa^n}{\sqrt{m!n!\theta^{m+n}}} \sqrt{\theta^q} \sqrt{\frac{(m+q)!}{m!}} \delta_{m+q,n}.\nonumber \\
&=&  \sum_{m} e^{-\frac{{\left\vert \kappa \right\vert }^2}{\theta}}  \frac{ {\left\vert \kappa \right\vert }^{2m} \bar \kappa^q}{\sqrt{m!(m+q)!\theta^{2m+q}}} \sqrt{\theta^q} \sqrt{\frac{(m+q)!}{m!}} \nonumber\\
&=&  \bar \kappa^q.
\end{eqnarray}
\end{proof}

From this proposition, coherent states can be evaluated on every linear combinaison of $z^q$ and $\bar z^p$, $p,q\in\mathbb{N}$, so polynomials belong to the unitization $\widetilde{\mathbb{A}}$.

We come now to the main theorem of this section, which shows that causal relations are possible between coherent states on Moyal plane:\\

\begin{theorem}\label{mainthmcausality}
Let us suppose that two coherent states $\omega_\xi,\omega_\varphi$ correspond to the complex scalars $\kappa_1,\kappa_2\in{\mathbb C}$. Those coherent states are causally related, with $\omega_\xi \preceq\omega_\varphi$, if and only if $\Delta\kappa:=\kappa_2-\kappa_1$ is inside the convex cone of ${\mathbb C}$ defined by $\lambda=\frac{1+i}{\sqrt2}$ and $\bar\lambda = \frac{1-i}{\sqrt2}$  (i.e.~the argument of $\Delta\kappa$ is within the interval $[-\frac{\pi}{4},\frac{\pi}{4}]$).\\
\end{theorem}

\begin{proof}
Let us first prove the sufficient condition, i.e.~we suppose that $\Delta\kappa=\kappa_2-\kappa_1 =  \mu\lambda + \nu\bar \lambda$ for some $\mu,\nu\geq0$. We want to show that $\forall a\in\mathcal{C}$, \mbox{$\omega_\varphi(a)-\omega_\xi(a) \geq 0$}. Using the definition of the coherent states in the matrix basis, this is equivalent to prove that:
\begin{equation}\label{goalcoherentstates}
\sum_{mn} (\bar \varphi_m \varphi_n - \bar \xi_m \xi_n) a_{mn} \geq 0
\end{equation}
for an arbitrary matrix $(a_{mn})_{mn\in{\mathbb N}} $ such that $(\alpha_{mn})_{mn\in{\mathbb N}}$ and $(\beta_{mn})_{mn\in{\mathbb N}}$, as defined by \eqref{expressalpha} and \eqref{expressbeta}, are  semi-positive definite matrices.

Let us define the following curve, with $t\in[0,1]$, within the set of pure states:
\begin{equation}  
\chi_m(t) := \frac{1}{\sqrt{2\pi\theta}}e^{-\frac{{\left\vert \kappa_1 \right\vert }^2}{2\theta}}  e^{-\frac{ \int_0^t ds\ \bar \kappa^\prime (s) \kappa(s)  }{\theta}} \frac{\kappa(t)^m}{\sqrt{m!\theta^m}},
\end{equation}
where $\kappa(t):= \kappa_1 + t(\kappa_2-\kappa_1) = \kappa_1 + t \Delta\kappa$ is the straight line between $\kappa_1$ and $\kappa_2$. This curve is normalized since:
\begin{equation}
 \sum_m {\left\vert  \chi_m(t) \right\vert }^2 =   \frac{1}{{2\pi\theta}} \sum_m e^{-\frac{{\left\vert \kappa_1 \right\vert }^2}{\theta}}  e^{-\frac{ \int_0^t ds\ (\bar \kappa^\prime (s) \kappa(s) +  \kappa^\prime (s) \bar \kappa(s)) }{\theta}} \frac{({\left\vert \kappa(t) \right\vert }^2)^m}{{m!\theta^m}} =  \frac{1}{{2\pi\theta}},
\end{equation}
with $ \int_0^t ds\ (\bar \kappa^\prime (s) \kappa(s) +  \kappa^\prime (s) \bar \kappa(s) ) =  \int_0^t  ds\ (\bar\kappa (s)  \kappa(s) )^\prime = { \left\vert \kappa(t) \right\vert }^2 - { \left\vert \kappa_1 \right\vert }^2$. Moreover:
\begin{equation} \bar \chi_m(0) \chi_n(0) =  \frac{1}{{2\pi\theta}}e^{-\frac{{\left\vert \kappa_1 \right\vert }^2}{\theta}}  \frac{\bar \kappa_1^m \kappa_1^n}{\sqrt{m!n!\theta^{m+n}}} = \bar \xi_m \xi_n,
\end{equation}
\begin{equation} \bar \chi_m(1) \chi_n(1) =  \frac{1}{{2\pi\theta}}e^{-\frac{{\left\vert  \kappa_1 \right\vert }^2}{\theta}}  e^{-\frac{{ \left\vert \kappa_2 \right\vert }^2 - { \left\vert \kappa_1 \right\vert }^2}{\theta}}  \frac{\bar \kappa_2^m \kappa_2^n}{\sqrt{m!n!\theta^{m+n}}} = \bar \varphi_m \varphi_n,
\end{equation}
so using the second fundamental theorem of calculus, our constraint \eqref{goalcoherentstates} becomes:
\begin{equation}\label{goal2coherentstates}
\sum_{mn} (\bar \chi_m(1) \chi_n(1)- \bar \chi_m(0) \chi_n(0)) a_{mn} = \int_0^1 dt\ \sum_{mn}  (\bar \chi_m(t) \chi_n(t))^\prime  a_{mn}  \geq 0.
\end{equation}

Since $(\alpha_{mn})_{mn\in{\mathbb N}}$ is semi-positive definite, by \eqref{expressalpha}  we obtain that, for every vector $\psi_m \in L^2({\mathbb N})$:
\begin{eqnarray}
&&\sum_{mn} \bar \psi_m \alpha_{mn} \psi_n \nonumber\\
 &=&  \frac{1}{\sqrt{\theta}}\sum_{mn} \bar \psi_m (a_{m+1,n} \lambda \sqrt{{m+1}{}} + a_{m,n+1} \bar\lambda \sqrt{{n+1}{}} - a_{m-1,n} \bar\lambda \sqrt{{m}{}} - a_{m,n-1} \lambda \sqrt{{n}{}}) \psi_n\nonumber\\
 &=&  \frac{1}{\sqrt{\theta}}\sum_{mn} \left[   ( \lambda \sqrt{m} \bar \psi_{m-1} - \bar\lambda \sqrt{m+1} \bar \psi_{m+1} ) \psi_n +  (\bar \lambda \sqrt{n}  \psi_{n-1} - \lambda \sqrt{n+1} \psi_{n+1} )  \bar \psi_m \right] a_{mn} \nonumber\\
 &\ \geq& 0.\quad \label{constfromalpha}
\end{eqnarray}

Doing the same with $(\beta_{mn})_{mn\in{\mathbb N}}$ and  \eqref{expressbeta}, we obtain:  

\begin{equation}\label{constfrombeta}
\frac{1}{\sqrt{\theta}}\sum_{mn} \left[   (\bar \lambda \sqrt{m} \bar \psi_{m-1} - \lambda \sqrt{m+1} \bar \psi_{m+1} )\psi_n +  ( \lambda \sqrt{n}  \psi_{n-1} - \bar\lambda \sqrt{n+1} \psi_{n+1} )  \bar \psi_m \right] a_{mn} \geq 0,
\end{equation}

and thus with the positive combination $ \mu \,\eqref{constfrombeta} + \nu \,\eqref{constfromalpha}$ we get:
\begin{eqnarray}
&\sum_{mn} &\left[   \left( \frac{\overline{\Delta\kappa}}{\sqrt{\theta}}   \sqrt{m} \bar \psi_{m-1} -  \frac{{\Delta\kappa}}{\sqrt{\theta}}  \sqrt{m+1} \bar \psi_{m+1} \right) \psi_n \right. \nonumber\\
&&+  \left. \left( \frac{{\Delta\kappa}}{\sqrt{\theta}} \sqrt{n}  \psi_{n-1} - \frac{\overline{\Delta\kappa}}{\sqrt{\theta}}   \sqrt{n+1} \psi_{n+1} \right)  \bar \psi_m \right] a_{mn} \geq 0. \label{constfromalphaandbeta}
\end{eqnarray}

Now for any given $t\in[0,1]$, we can choose $\psi_m = \chi_m(t)$. We have that:
\begin{eqnarray}
 \chi^\prime_m(t) &=& \frac{1}{\sqrt{2\pi\theta}} e^{-\frac{{ \left\vert \kappa_1 \right\vert }^2}{2\theta}}  e^{-\frac{ \int_0^t ds\ \bar \kappa^\prime (s) \kappa(s) }{\theta}} \frac{\kappa(t)^{m-1}}{\sqrt{m!\theta^m}} \;m\; \kappa^\prime(t) \nonumber\\
 &&- \frac{1}{\sqrt{2\pi\theta}} e^{-\frac{{ \left\vert \kappa_1 \right\vert }^2}{2\theta}}  e^{-\frac{ \int_0^t ds\ \bar \kappa^\prime (s) \kappa(s) }{\theta}} \frac{\kappa(t)^{m+1}}{\sqrt{m!\theta^m}} \frac{\bar\kappa^\prime(t)}{\theta}\nonumber\\
 &=& \frac{\Delta\kappa}{\sqrt{\theta}} \sqrt{m} \chi_{m-1}(t) -   \frac{\overline{\Delta\kappa}}{\sqrt{\theta}}  \sqrt{m+1} \chi_{m+1}(t) \nonumber\\
 &=& \frac{\Delta\kappa}{\sqrt{\theta}} \sqrt{m} \psi_{m-1} -   \frac{\overline{\Delta\kappa}}{\sqrt{\theta}}  \sqrt{m+1} \psi_{m+1}.
\end{eqnarray}

So the inequality \eqref{constfromalphaandbeta} gives us, for every $t\in[0,1]$:
\begin{equation}  
\sum_{mn} \left[  \bar \chi^\prime_m(t) \chi_n(t) +   \chi^\prime_n(t) \bar \chi_m(t) \right] a_{mn} = \sum_{mn} (\bar \chi_m(t) \chi_n(t))^\prime  a_{mn} \geq 0,
\end{equation}
which implies the constraint \eqref{goal2coherentstates}. Hence the sufficient condition is proved.\\

Now we turn into the necessary condition. We want to show that, when the argument of $\Delta\kappa$ is not in the interval $[-\frac{\pi}{4},\frac{\pi}{4}]$, then $\omega_\xi \not\preceq\omega_\varphi$, which means that there exists an element in $\mathcal{C}$ such that $\omega_\varphi(a)-\omega_\xi(a) < 0$.

Let us consider the function $a := \lambda z + \bar \lambda \bar z$, with $z=\frac{x_0+ix_1}{\sqrt 2}$. Using the same computation than in Proposition \ref{coherentstatesarefinite} and linearity, we have $a = \sum_{mn} a_{mn} f_{mn} $ with
\begin{eqnarray}
a_{mn}  &=&  \frac{\lambda}{2\pi\theta} \left<z,f_{mn}\right> +  \frac{\bar\lambda}{2\pi\theta} \left<\bar z,f_{mn}\right> \nonumber\\
&=&  \lambda\sqrt{\theta} \sqrt{m+1} \delta_{m+1,n}   +   \bar \lambda\sqrt{\theta} \sqrt{n+1} \delta_{m,n+1},
\end{eqnarray}
and
\begin{eqnarray}
&&\omega_\varphi(a)-\omega_\xi(a)\nonumber\\
&=&2\pi\theta \sum_{mn} (\bar \varphi_m \varphi_n - \bar \xi_m \xi_n) a_{mn}\nonumber\\
&=&  \sum_{mn}  \left(e^{-\frac{{\left\vert  \kappa_2 \right\vert }^2}{\theta}}  \frac{\bar \kappa_2^m \kappa_2^n}{\sqrt{m!n!\theta^{m+n}}} - e^{-\frac{{ \left\vert \kappa_1 \right\vert }^2}{\theta}}  \frac{\bar \kappa_1^m \kappa_1^n}{\sqrt{m!n!\theta^{m+n}}}  \right) a_{mn} \nonumber\\
&=&  \sum_{m}  \left(e^{-\frac{{\left\vert \kappa_2 \right\vert }^2}{\theta}}  \frac{ {\left\vert \kappa_2 \right\vert }^{2m} \kappa_2}{\sqrt{m!(m+1)!\theta^{2m+1}}} - e^{-\frac{{ \left\vert \kappa_1 \right\vert }^2}{\theta}}  \frac{ {\left\vert  \kappa_1 \right\vert }^{2m} \kappa_1}{\sqrt{m!(m+1)!\theta^{2m+1}}}  \right) \lambda\sqrt{\theta} \sqrt{m+1} \nonumber\\
&+&  \sum_{m}  \left(e^{-\frac{{\left\vert \kappa_2 \right\vert }^2}{\theta}}  \frac{ {\left\vert  \kappa_2 \right\vert }^{2m} \bar\kappa_2}{\sqrt{m!(m+1)!\theta^{2m+1}}} - e^{-\frac{{\left\vert \kappa_1 \right\vert }^2}{\theta}}  \frac{ { \left\vert \kappa_1 \right\vert }^{2m} \bar\kappa_1}{\sqrt{m!(m+1)!\theta^{2m+1}}}  \right) \bar\lambda\sqrt{\theta} \sqrt{m+1} \nonumber\\
&=&  \sum_{m} \left(e^{-\frac{{\left\vert  \kappa_2 \right\vert }^2}{\theta}}  \frac{ {\left\vert  \kappa_2 \right\vert }^{2m} }{m!\theta^m} (\lambda \kappa_2 + \bar\lambda\bar\kappa_2) - e^{-\frac{{\left\vert  \kappa_1 \right\vert }^2}{\theta}}  \frac{ { \left\vert \kappa_1 \right\vert }^{2m} }{m!\theta^m} (\lambda \kappa_1 + \bar\lambda\bar\kappa_1) \right)\nonumber\\
&=&  \lambda \Delta\kappa + \bar\lambda\overline{\Delta\kappa} = {\left\vert \Delta\kappa \right\vert } e^{\frac{\pi}{4}+\theta_{\Delta\kappa}} + {\left\vert \Delta\kappa \right\vert }e^{-\frac{\pi}{4}-\theta_{\Delta\kappa}}
= 2 {\left\vert \Delta\kappa \right\vert } \cos\left(\frac{\pi}{4}+\theta_{\Delta\kappa}\right)\nonumber\\
\phantom{X}
\end{eqnarray}
where $\Delta\kappa = {\left\vert \Delta\kappa \right\vert } e^{\theta_{\Delta\kappa}}$ with $\theta_{\Delta\kappa} \in ]-\pi,\pi]$. This is negative when $\theta_{\Delta\kappa} \notin [-\frac{3\pi}{4},\frac{\pi}{4}]$.

If we do the same with $\tilde a := \bar \lambda z +  \lambda \bar z$, we find:
\begin{eqnarray}
\omega_\varphi(\tilde a)-\omega_\xi(\tilde a) &= & \bar\lambda \Delta\kappa + \lambda\overline{\Delta\kappa} \nonumber\\
&=& {\left\vert \Delta\kappa \right\vert }e^{-\frac{\pi}{4}+\theta_{\Delta\kappa}} +{\left\vert \Delta\kappa \right\vert } e^{\frac{\pi}{4}-\theta_{\Delta\kappa}} \nonumber\\
&=& 2 {\left\vert \Delta\kappa \right\vert }\cos\left(-\frac{\pi}{4}+\theta_{\Delta\kappa}\right),
\end{eqnarray} 
which is negative when $\theta_{\Delta\kappa} \notin [-\frac{\pi}{4},\frac{3\pi}{4}]$. Hence, if $\theta_{\Delta\kappa} \notin [-\frac{\pi}{4},\frac{\pi}{4}]$, one of the functions $a$ or $\tilde a$ contradicts the causality condition.

In order to finish the proof, we need to verify that $a$ and $\tilde a$ are in $\mathcal{C}$, i.e.~that $(\alpha_{mn})_{mn\in{\mathbb N}}$, $(\beta_{mn})_{mn\in{\mathbb N}}$, $(\tilde \alpha_{mn})_{mn\in{\mathbb N}}$ and $(\tilde \beta_{mn})_{mn\in{\mathbb N}}$ are  semi-positive definite matrices.

We have:
\begin{eqnarray}
\alpha_{mn} =\ \tilde \beta_{mn} &=& \lambda^2 \sqrt{m+1}\sqrt{m+2} \,\delta_{m+2,n} + {\left\vert \lambda \right\vert }^2 (m+1) \,\delta_{m,n} \nonumber\\
&& +  {\left\vert \lambda \right\vert }^2 (m+1) \,\delta_{m,n} + \bar  \lambda^2 \sqrt{m-1}\sqrt{m}  \,\delta_{m,n+2} \nonumber\\
&& -  {\left\vert \lambda \right\vert }^2 (m) \,\delta_{m,n} - \bar  \lambda^2 \sqrt{m}\sqrt{m-1}  \,\delta_{m,n+2} \nonumber\\
&& -   \lambda^2 \sqrt{m+2}\sqrt{m+1} \,\delta_{m+2,n} - {\left\vert \lambda \right\vert }^2 (m)  \,\delta_{m,n} \nonumber\\
&=& 2 {\left\vert \lambda \right\vert }^2  \,\delta_{m,n} = 2 \,\delta_{m,n} ,
\end{eqnarray}
which is clearly  semi-positive definite, and 
\begin{eqnarray}
\beta_{mn} =\ \tilde \alpha_{mn} &=& {\left\vert \lambda \right\vert }^2 \sqrt{m+1}\sqrt{m+2} \,\delta_{m+2,n} + \bar \lambda^2 (m+1) \,\delta_{m,n} \nonumber\\
&& +  \lambda^2 (m+1) \,\delta_{m,n} + {\left\vert \lambda \right\vert }^2 \sqrt{m-1}\sqrt{m}  \,\delta_{m,n+2} \nonumber\\
&& -  \lambda^2 (m) \,\delta_{m,n} - {\left\vert \lambda \right\vert }^2 \sqrt{m}\sqrt{m-1}  \,\delta_{m,n+2} \nonumber\\
&& -   {\left\vert \lambda \right\vert }^2 \sqrt{m+2}\sqrt{m+1} \,\delta_{m+2,n} - \bar {\lambda}^2 (m)  \,\delta_{m,n} \nonumber\\
&=& ( \lambda^2 +  \bar  \lambda^2)  \,\delta_{m,n} = 0,
\end{eqnarray}
which is the null matrix.
\end{proof}

Theorem \ref{mainthmcausality} explicitly shows us that causal relations between some pure states on Moyal place are possible. Surprisingly, the causal structure between the coherent states mimics the causal structure on the usual Minkowski space, using future (and past) cones of light defined on $\mathbb C \cong {\mathbb R}^{1,1}$ and with the identification of the real line of $\mathbb C$ with the time line of ${\mathbb R}^{1,1}$. The causality between Gaussian functions centered on specific points corresponds exactly to the usual causality between those points, if we interpret the coherent states as translations of the ground state $f_{00}$. Since in quantum mechanics, coherent states are exactly the states that are expected to behave as classically as possible, the obtained result, which is the first explicit causal relation discovered on Moyal plane, clearly reaches this expectation. However, coherent states only represent a small part of the possible states on Moyal plane, so the complete causal structure could be far richer. The possible causal relations between generalized states will be the subject of a future investigation.

\end{document}